\documentclass{llncs}
\pdfoutput=1

\usepackage{amsmath,amssymb}
\usepackage{graphicx,color}
\usepackage{mathpartir}
\usepackage{enumitem}
\usepackage{textcomp}
\usepackage{hyperref}

\newcommand{\ignore}[1]{}

\newcommand{\eqdef}{\stackrel{def}=}
\newcommand{\predname}[1]{\mbox{\bf #1}}

\newcommand{\Reals}{\mathbb{R}}
\newcommand{\PosReals}{\mathbb{R}_{\ge 0}}

\newcommand{\norm}[1]{{|#1|}}

\newcommand{\Segments}{{\cal S}}

\newcommand{\segnorm}[1]{{|\!|#1|\!|}}
\newcommand{\subseg}[3]{{#1}[{#2},{#3}]}

\newcommand{\pre}[1]{^\bullet\!{#1}}
\newcommand{\post}[1]{{#1}^\bullet}

\newcommand{\segment}[1]{{#1}.s}
\newcommand{\Pos}[1]{\mathit{Pos}_{#1}}
\newcommand{\pos}[2]{\mathit{pos}_{#1}^{#2}}
\newcommand{\gedge}[4]{{#2} \xrightarrow{#3}_{#1} {#4} }
\newcommand{\gtran}[4]{{#2} \stackrel{#3}\leadsto_{#1} {#4} }
\newcommand{\distance}[1]{d_{#1}}

\newcommand{\refines}{ \sqsubseteq }

\newcommand{\closure}[1]{\sim\!\!{#1}}

\newcommand{\MCL}{$\mathit{MCL}$}
\newcommand{\SL}{$\mathit{SL}$}
\newcommand{\MMCL}{$\mathit{M2CL}$}
\newcommand{\TMMCL}{$\mathit{TM2CL}$}

\renewcommand{\vec}[1]{\mathbf{#1}}

\newcommand{\att}[2]{#1.\mathit{#2}}
\newcommand{\Vehicles}{{\cal C}}
\newcommand{\Objects}{{\cal O}}

\newcommand{\meets}[3]{ [{#1}~\mathit{meets}(#2)~#3] }
\newcommand{\inside}[2]{ [{#1} \mathit{@} {#2}] }

\newcommand{\always}{\Box}
\newcommand{\eventually}{\Diamond}
\newcommand{\next}{\mathbf{N}~}
\newcommand{\until}{~\mathbf{U}~}

\newcommand{\extname}[1]{\mbox{\it #1}}

\title{Specification and Validation of Autonomous Driving Systems:
  A Multilevel Semantic Framework}

\author{Marius Bozga \and Joseph Sifakis}

\institute{Univ. Grenoble Alpes, CNRS, Grenoble INP\footnote{Institute of Engineering Univ. Grenoble Alpes}, VERIMAG}

\begin{document}
\maketitle

\begin{abstract}
  
Autonomous Driving Systems (ADS) are critical dynamic reconfigurable agent
systems whose specification and validation raises extremely challenging
problems. The paper presents a multilevel semantic framework for the
specification of ADS and discusses associated validation problems. The
framework relies on a formal definition of maps modeling the physical
environment in which vehicles evolve. Maps are directed metric graphs whose
nodes represent positions and edges represent segments of roads. We study
basic properties of maps including their geometric
consistency. Furthermore, we study position refinement and segment
abstraction relations allowing multilevel representation from purely
topological to detailed geometric. We progressively define first order
logics for modeling families of maps and distributions of vehicles over
maps. These are Configuration Logics, which in addition to the usual
logical connectives are equipped with a coalescing operator to build
configurations of models. We study their semantics and basic properties. We
illustrate their use for the specification of traffic rules and scenarios
characterizing sequences of scenes.  We study various aspects of the
validation problem including run-time verification and satisfiability of
specifications. Finally, we show links of our framework with practical
validation needs for ADS and advocate its adequacy for addressing the many
facets of this challenge.

\end{abstract}


\section{Introduction}
The validation of ADS raises challenges far beyond the current state of the
art because of their overwhelming complexity and the integration of
non-explainable AI components.  Providing sufficient evidence that these
systems are safe enough is a hot and critical need, given the underlying
economic and societal stakes. This objective mobilizes considerable
investments and efforts by key players including big tech companies and car
manufacturers. The efforts focus on the development of efficient simulation
technology and common infrastructure for modelling the physical environment
of ADS and their desired properties. They led in particular to the
definition of common formats such as
OpenDRIVE \cite{OpenDRIVE-1.4,ASAMOpenDRIVE-1.6.0} for the description of
road networks, and OpenSCENARIO \cite{ASAMOpenScenario-1.0.0} for the
description of complex, synchronized maneuvers that involve multiple
entities like vehicles, pedestrians and other traffic
participants. Additionally, several open simulation environments such as
CARLA \cite{DosovitskiyRCLK17} and LGSVL \cite{abs-2005-03778} are
available for modelling and validation.

As a rule, existing industrial simulation environments are built on top of
game engines. They privilege realism of modeling but allow poor semantic
awareness. Simulation tools are not rooted in rigorous semantics that could
provide a basis for analysis and reasoning about model properties. For
instance, they provide no support to check that a map obtained by
composition of road components e.g. using OpenDRIVE in CARLA, is consistent
with its geometric interpretation. Furthermore, lack of semantic awareness
makes impossible the development of convincing technical arguments
regarding the safety of the simulated systems. Accident-free simulation
even for billions of miles is not a conclusive safety evidence if it is not
possible to explain how simulated miles are related to real
miles. Obviously, any technically sound safety evaluation should be
model-based and this is not possible under the current state of the art.

The paper proposes a semantic framework for the specification and
validation of ADS. The framework provides a precise semantic model of the
environment of ADS based on maps. It also includes logics for the
specification and validation of properties of the semantic model and of the
system dynamic behavior.

Maps have been the object of numerous studies focusing on the formalization
of the concept and its use for the analysis of ADS.  A key research issue
is to avoid monolithic representations and build maps by composition of
components and heterogeneous data. This motivated formalizations using
ontologies and logics with associated reasoning mechanisms to check
consistency of descriptions and their correctness with respect to desired
properties \cite{BeetzB18,BagschikMM18} or to generate
scenarios \cite{BagschikMM18,ChenK18}. Other works propose open source map
frameworks for highly automated driving
\cite{OpenDRIVE-1.4,ASAMOpenDRIVE-1.6.0,PoggenhansPJONK18}.

A different research line focuses on the validation of ADS either to verify
satisfaction of safety and efficiency properties or even to check that
vehicles respect given traffic rules.  Many works deal with safety
verification in a simple multilane setting. In \cite{HilscherLOR11} a
dedicated Multi-Lane Spatial Logic inspired by interval temporal logic is
used to specify safety and provide proofs for lane change controllers. The
work in \cite{RizaldiISA18} presents a motion planner formally verified in
Isabelle/HOL. The planner is based on manoeuver automata, a variant of
hybrid automata, and properties are expressed in linear temporal logic.

Other works deal with scenarios for modeling the behavior of
ADS. OpenSCENARIO \cite{ASAMOpenScenario-1.0.0} defines a data model and a
derived file format for the description of scenarios used in driving and
traffic simulators, as well as in automotive virtual development, testing
and validation.  The work in \cite{DammKMPR18} proposes a visual formal
specification language for capturing scenarios inspired from Message Charts
and shows possible applications to specification and testing of autonomous
vehicles. In \cite{SchonemannWGO+19} a scenario-based methodology for
functional safety analysis is presented using the example of automated
valet parking. The work in \cite{FremontKPSABWLL20} presents an approach to
automated scenario-based testing of the safety of autonomous vehicles,
based on Metric Temporal Logic.  Finally, the probabilistic language Scenic
for the design and analysis of cyber physical systems allows the
description of scenarios used to control and validate simulated systems of
self-driving cars. The Scenic programming environment provides a big
variety of constructs making possible modeling anywhere in the spectrum
from concrete scenes to broad classes of abstract scenarios
\cite{abs-1809-09310,abs-2010-06580}.

Other works focus on checking compliance of vehicles with traffic rules. A
formalization of traffic rules in linear temporal logic is proposed
in \cite{EsterleAK19}. Runtime verification is applied to check that
maneuvers of a high-level planner comply with the rules. Works in
\cite{RizaldiKHFIAHN17,RizaldiA15} formalize a set of traffic rules for
highway scenarios in Isabelle/HOL; they show that traffic rules can be used
as requirements to be met by autonomous vehicles and propose a verification
procedure. A formalization of traffic rules for uncontrolled intersections
is provided in \cite{KarimiD20} using the CLINGO logic programming
language. Furthermore, the rules are applied by a simulator to safely
control traffic across intersections. The work in \cite{EsterleGK20}
proposes a methodology for the formalization of traffic rules in Linear
Temporal Logic; it is shown how evaluation of formalized rules on recorded
drives of humans provides insight on what extent drivers respect the rules.

This work is an attempt to provide a minimal framework unifying the
concepts for the specification of ADS and the associated validation
problems. The proposed semantic framework clearly distinguishes between a
static part consisting of the road network with its equipment and a dynamic
part involving objects.  The static part is a map described as a metric
graph obtained as the composition of subgraphs representing roads and
junctions at different abstraction levels. The vertices of the graph are
positions and its edges are road segments. Depending on the chosen level of
abstraction, segments can be simply the distance between the connected
positions or curves or even two-dimensional regions. The geometric
interpretation of segments implies that maps should meet consistency
properties studied in the paper.  It allows multilevel representation using
position refinement and segment abstraction. The proposed concept of map is
quite general encompassing compositional multilevel description of traffic
networks.  The dynamic part of the ADS model is a transition system whose
state characterizes the distribution of objects over a map with their
positions, itineraries and attributes of their kinematic state.

We progressively introduce three logics to express properties of the
semantic model at different levels.  The \emph{Metric Configuration Logic}
(\MCL) allows the compositional and parametric description of metric
graphs. This is a first order logic with variables ranging over positions
and segments. It uses in addition to logical connectives, a coalescing
operator for the compositional construction of maps from segments. A \MCL\
formula represents configurations of maps sharing a common set of
locations. We discuss a specification methodology and show how various road
patterns such as roundabouts, intersections, mergers of roads can be
specified in \MCL.

The Mobile Metric Configuration Logic (abbreviated \MMCL) is an extension
of \MCL\ with additional object variables and primitives for the
specification of scenes as the distribution of objects over maps. \MMCL\ 
formulas can be written as the conjunction of formulas describing: i)
static map contexts; ii) dynamic relations between objects; iii) addressing
relations between objects and maps.  Last, we define Temporal \MMCL\ 
(abbreviated \TMMCL), a linear temporal logic whose atomic propositions are
formulas of \MMCL. We illustrate the use of these logics for the
specification of safety properties including traffic rules as well as the
description of dynamic scenarios.

Additionally, we study the validation of properties expressed in the three
logics and provide a classification of problems showing that validation of
general dynamic properties boils down to constraint checking on metric
graphs.  Checking that a finite model satisfies a formula of \MCL\
or \MMCL\ amounts to eliminate quantifiers by adequate instantiation of
variables. We argue that satisfiability of \MMCL\ formulas can be reduced
to satisfiability of
\MCL\ formulas which is an undecidable problem. We identify a reasonably
expressive decidable subset of \MCL\ and propose a decision
procedure. Furthermore, we discuss the problem of runtime verification
of \TMMCL\ formulas and sketch a principle of solution inspired from a
recent work with a similar configuration logic \cite{El-HokayemBS21}.  We
complete the presentation on ADS validation with an analysis of
practical needs for a rigorous validation methodology. We describe a
general validation environment and show how the proposed framework provides
insight into the different aspects of validation and related methodological
issues.

The proposed framework allows a holistic treatment of all the aspects of
ADS modeling. It is not limited to specific contexts such as simple
multilane setting, intersection, roundabout, parking, etc. It fully
encompasses both dynamic and static aspects. Finally, the proposed concept
of map allows conciseness and precision not achievable by ontologies where
semantic issues are often overloaded and obscured by details that can be
added to our model without affecting its basic properties.

The paper is structured as follows. In section \ref{sec:mcl}, we study
metric graphs and their relevant properties for the representation of map
models as well as the logic \MCL, its main properties and application for
map specification. Section \ref{sec:ads-specification} deals with the study
of logics \MMCL\ and \TMMCL\ and their application to the specification of
safety properties and the description of scenarios. Then, section
\ref{sec:ads-validation} discusses a classification of validations
problems, approaches for their solution and their effective
application. Section \ref{sec:discussion} concludes with a summary of main
results and a discussion about future developments.

\section{Metric Graphs and Metric Configuration Logic} \label{sec:mcl}

\subsection{Segments and Metric Graphs}

\subsubsection{Segments.}

We build contiguous road segments from a set $\Segments$ equipped
with a partial concatenation operator $\cdot : \Segments \times
\Segments{} \rightarrow \Segments \cup \{\bot\}$ and a length norm $\segnorm{.} :
\Segments \rightarrow \PosReals$ satisfying the following properties:
\begin{enumerate}[label=(\roman*)]
\item \emph{associativity:} for any segments $s_1$, $s_2$, $s_3$ either both $(s_1
  \cdot s_2) \cdot s_3$ and $s_1 \cdot (s_2 \cdot s_3)$ are defined and equal, or both undefined;
\item \emph{length additivity wrt concatenation:} for any segments $s_1$, $s_2$
  whenever $s_1 \cdot s_2$ defined it holds $\segnorm{s_1 \cdot s_2}
  = \segnorm{s_1} + \segnorm{s_2}$;
\item \emph{segment split:} for any segment $s$ and non-negative $a_1$, $a_2$
  such that $\segnorm{s} = a_1 + a_2$ there exist unique $s_1$, $s_2$
  such that $s = s_1
  \cdot s_2$, $\segnorm{s_1} = a_1$, $\segnorm{s_2} = a_2$. 
\end{enumerate}
The last property allows us to define consistently a subsegment
operation: $\subseg{s}{a_1}{a_2}$ is the unique segment of length $a_2
- a_1$ satisfying $s = s_1 \cdot \subseg{s}{a_1}{a_2} \cdot s_2$ where
$s_1$, $s_2$ are such that $\segnorm{s_1} = a_1$, $\segnorm{s_2} =
\segnorm{s} - a_2$, for any $0 \le a_1 \le a_2 \le \segnorm{s}$.  For
brevity, we use the shorthand notation $\subseg{s}{a}{\textrm{-}}$
to denote the subsegment $\subseg{s}{a}{\segnorm{s}}$.

Segments will be used to model building blocks of roads in maps
considering three different interpretations. Interval segments simply
define the length of a segment. Curve segments define the precise
geometric form of the trajectory of a mobile object along the
segment. Region segments are 2D-regions of given width around a center
curve segment.

\paragraph{Interval Segments.}

Consider $\Segments_{interval} \eqdef \{ [0,a] ~|~ a \in \PosReals \}$,
that is, the set of closed intervals on reals with lower bound 0,
concatenation defined by $[0, a_1] \cdot [0, a_2] \eqdef [0, a_1 + a_2]$
and length $\segnorm{ [0, a] } \eqdef a$.

\paragraph{Curve Segments.}

Consider $\Segments_{curve} \eqdef \{ c : [0,1] \rightarrow \Reals^2 ~|~
c(0) = (0,0),~ c \mbox{ curve} \} \cup \{ \epsilon \}$ that is, the set of
curves that are continuous smooth\footnote{the derivative $\dot{c}$ exists
  and is continuous on $[0,1]$} and uniformly progressing\footnote{the
  instantaneous speed $\norm{\dot{c}}$, that is, the Euclidean norm of the
  derivative is constant} functions $c$, starting at the origin, plus a
designated single point curve $\epsilon$.  The length is defined by taking
respectively the length of the curve $\segnorm{c} \eqdef \int_0^1
\norm{\dot{c}(t)} dt$ and $\segnorm{\epsilon} = 0$.  The concatenation $c_1
\cdot c_2$ of two curves $c_1$ and $c_2$ is a partial operation that
consists in joining the final endpoint of $c_1$ with the initial endpoint
of $c_2$ provided the slopes at these points are equal. This condition
preserves smoothness of the curve $c_1 \cdot c_2$ defined by $c_1 \cdot
c_2: [0,1] \rightarrow \Reals^2$ where:
  $$(c_1 \cdot c_2)(t) \eqdef \left\{ \begin{array}{rcl}
    c_1(\frac{t}{\lambda})& \mbox{if}
    & t \in [0,\lambda] \\[3pt]
    c_1(1) + c_2(\frac{t-\lambda}{1 - \lambda}) & \mbox{if} & t
    \in [\lambda, 1] \end{array}\right.
\mbox{ where } \lambda=\frac{\segnorm{c_1}}{\segnorm{c_1} + \segnorm{c_2}}$$
Note that in this definition, $c_1$ and $c_2$ are scaled on
sub-intervals of $[0,1]$ respecting their length ratio.  We
additionally take $c \cdot \epsilon \eqdef \epsilon \cdot c \eqdef c$, for any
$c$.  For practical reasons, one can further restrict the set
$\Segments_{curve}$ to curves of some form e.g, finite concatenation of
parametric line segments and circle arcs. That is, for any $a, r\in
\PosReals^*$, $\varphi\in \Reals$, $\theta \in \Reals^*$ the curves
$line[a,\varphi]$, $arc[r,\varphi,\theta]$ are defined as
  $$\begin{array}{rcl}
    line[a,\varphi](t) & \eqdef & (at\cos\varphi,at\sin\varphi) ~~\forall t\in[0,1]\\
    arc[r,\varphi,\theta](t) & \eqdef & (r(\sin (\varphi + t\theta) - \sin\varphi),
    r (-\cos(\varphi + t \theta) + \cos\varphi)) ~~\forall t\in [0,1] 
  \end{array}$$
Note that $a$ and $r$ are respectively the length of the line and the
radius of the arc, $\varphi$ is the slope of the curve at the initial
endpoint and $\theta$ is the degree of the arc.
Fig.~\ref{fig:ex-curves} illustrates the composition of three segments
of this parametric form.

\begin{figure}[htbp]
  \begin{center}
    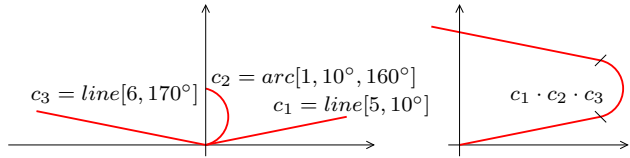
  \end{center}
  \caption{\label{fig:ex-curves}Curve segments and their composition}
\end{figure}

\paragraph{Region Segments.}

Consider $\Segments_{region} \eqdef \Segments_{curve} \times \PosReals^*$,
that is, the set of pairs $(c,w)$ where $c$ is a curve and $w$ a positive
number, denoting respectively the region center curve and the region
width. Region segments can be concatenated iff their curves can be
concatenated and if their widths are equal, that is, $(c_1, w) \cdot (c_2,
w) \eqdef (c_1 \cdot c_2, w)$ if $c_1 \cdot c_2 \not= \bot$.  The length of
a region segment is defined as the length of its center curve,
$\segnorm{(c, w)} \eqdef \segnorm{c}$.

Region segments can be equally understood as subsets of points in
$\Reals^2$ defined by algebraic constraints.  More precisely, for any curve
$c$ and width $w$ the region segment $(c,w)$ corresponds to the subset of
$\Reals^2$ defined as $\{c(t) + \lambda \cdot \frac{\mbox{\small \it
    ortho}(\dot{c}(t))}{ \norm{\dot{c}(t)}} ~|~ t\in [0,1],~ \lambda \in
[-\frac{w}{2}, \frac{w}{2}]\}$ where $\mbox{\it ortho}$ is the orthogonal
operator on $\Reals^2$ defined as $ortho((a,b)) \eqdef (-b,a)$.  In
particular, the above definition allows us to systematically derive 
parametric characterizations for region segments constructed using line or
arc curves. The region generated by the curve $line[a,\varphi]$ is a
rectangle containing the set of points $\{ (at\cos\varphi
-\lambda\sin\varphi,at\sin\varphi + \lambda\cos\varphi) ~|~ t \in [0,1],
\lambda\in[-\frac{w}{2},\frac{w}{2}]\}$.  The region generated by the curve
$arc[r,\varphi,\theta]$ is a ring sector containing the set of points $\{
((r+\lambda)(\sin(\varphi+t\theta)-r\sin\varphi,-(r+\lambda)\cos(\varphi+t\theta)
+ r\cos\varphi) ~|~ t\in [0,1], \lambda\in[-\frac{w}{2},\frac{w}{2}]\}$.

\subsubsection{Metric Graphs.}

We use metric graphs $G \eqdef (V,\Segments,E)$ to represent maps, where
$V$ is a finite set of \emph{vertices}, $\Segments$ is a set of segments
and $E \subseteq V \times \Segments^\star \times V$ is a finite set of
\emph{edges} labeled by non-zero length segments in $\Segments^\star$.  We
also denote an edge $e=(v,s,v') \in E$ by $\gedge{G}{v}{s}{v'}$ and we
define $\pre{e} \eqdef v$, $\post{e} \eqdef v'$, $\segment{e} \eqdef s$.
For a vertex $v$, we define $\pre{v} \eqdef \{ e | \post{e} = v\}$ and
$\post{v} \eqdef \{e | \pre{e} = v\}$.  We denote by $E^+_{ac}$ the finite
set of non-empty \emph{acyclic\footnote{every edge occurs at most once in
    the path}} directed paths with edges from $E$.  We call a metric graph
\emph{strongly} (resp. \emph{weakly}) connected if a \emph{directed}
(resp. \emph{undirected}) path exists between any pair of vertices.  A
metric graph is called \emph{acyclic} if at most one path, directed or
undirected, exist between any pairs of vertices.


We consider the set $\Pos{G} \eqdef V \cup \{(e,a) ~|~ e \in E,~ 0 < a
< \segnorm{\segment{e}}\}$ of \emph{positions} defined by a metric
graph.  Note that $(e,0)$ and $(e,
\segnorm{\segment{e}})$ are respectively the positions $\pre{e}$ and $\post{e}$.
Moreover, a $s$-labelled \emph{ride} between positions
$(e,a)$ and $(e',a')$ is an acyclic path denoted by $\gtran{G}{(e,a)}{s}{(e',a')}$ and
defined as follows:

\begin{enumerate}[label=(\roman*)]
\item $e = e'$, $0 \le a \le a' \le \segnorm{\segment{e}}$, $s = \subseg{\segment{e}}{a}{a'}$
\item $e = e'$, $0 \le a' \le a \le \segnorm{\segment{e}}$, $\post{e} = \pre{e}$, $s =
  \subseg{\segment{e}}{a}{\textrm{-}} \cdot
  \subseg{\segment{e}}{0}{a'} \not=\bot$
\item $e = e'$, $0 \le a' \le a \le \segnorm{\segment{e}}$, $w\in E^+_{ac}$, $e \not\in w$,
  $\post{e} = \pre{w}$, $\post{w} = \pre{e}$,  \\ $s =
  \subseg{\segment{e}}{a}{\textrm{-}} \cdot
  \segment{w} \cdot 
  \subseg{\segment{e}}{0}{a'}\not=\bot$
\item $e \not= e'$, $\post{e} = \pre{e'}$, $s = \subseg{\segment{e}}{a}{\textrm{-}} \cdot \subseg{\segment{e'}}{0}{a'} \not=\bot$
\item $e \not= e'$, $w\in E^+_{ac}$, $e,e' \not\in w$, $\post{e} = \pre{w}$, $\post{w} = \pre{e'}$,
  $s = \subseg{\segment{e}}{a}{\textrm{-}} \cdot  \segment{w} \cdot \subseg{\segment{e'}}{0}{a'}\not=\bot$
\end{enumerate}

Fig.~\ref{fig:acyclic-rides} illustrates the five cases of the
above definition for a simple graph with segments $s_1$, $s_2$ and
$s_3$. Cases (i) and (ii) correspond to rides on the same
segment. Case (iii) corresponds to rides originating and
terminating in fragments of the same segment and also involving other
segments between them.  Finally cases (iv) and (v) are rides
originating and terminating at different segments.

\begin{figure}[htbp]
  \begin{center}
    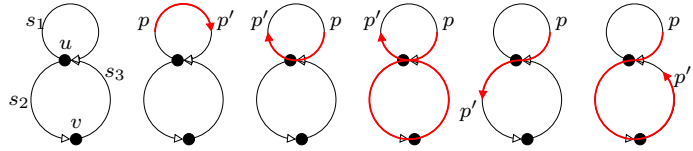
    \caption{\label{fig:acyclic-rides} Rides in metric graphs - cases (i)-(v) illustrated}
  \end{center}
\end{figure}

We define the distance $\distance{G}$ between positions $p$, $p'$ as 0
whenever $p = p'$ or the minimum length among all segments labeling rides
from $p$ to $p'$ and otherwise $+\infty$ if no such ride exists.  It can be
checked that $\distance{G}$ is an \emph{extended quasi-metric} on the set
$\Pos{G}$ and therefore, $(\Pos{G}, \distance{G})$ is an extended
quasi-metric space.

\subsection{Properties of Metric Graphs}

\subsubsection{Contraction/Refinement.}

A metric graph $G' = (V',\Segments,E')$ is a \emph{contraction} of a
metric graph $G = (V,\Segments,E)$ (or dually, $G$ is a
\emph{refinement} of $G'$), denoted by $G \refines G'$, iff $G$ is
obtained from $G'$ by transformations replacing some of its edges $e$
by acyclic sequences of interconnected edges $e_1 e_2 ... e_n$ while
preserving the segment labeling i.e., $\segment{e} = \segment{e_1}
\cdot \segment{e_2} \cdot ... \cdot \segment{e_n}$.

\begin{example}
  In Fig.~\ref{fig:metric-graphs}, the graph on the right is a
  contraction of the one on the left iff $s_{12} = s_{14} \cdot s_{45}
  \cdot s_{52}$, $s'_{12} = s'_{16} \cdot s'_{62}$ and $s_{31} =
  s_{37} \cdot s_{78} \cdot s_{81}$.
\end{example}

\begin{figure}[htbp]
  \begin{center}
    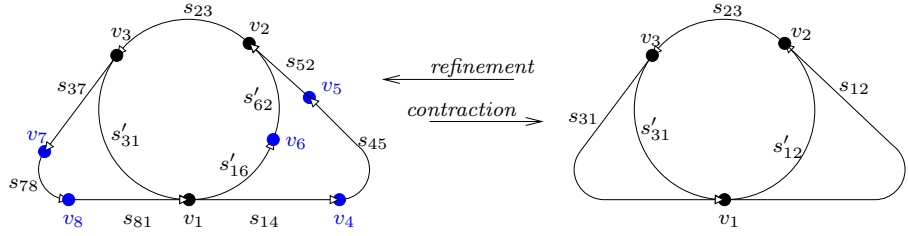
    \caption{\label{fig:metric-graphs} Illustration of
      contraction/refinement on metric graphs}
  \end{center}
\end{figure}

Note that metric graphs where all vertices have input or output degree
greater than one cannot be contracted.  Such vertices correspond
to \emph{junctions} (confluence of divergence of roads) when metric graphs
represent maps.  The following proposition states some key properties on
contraction/refinement of metric graphs.

\begin{proposition} \label{prop:position-refinement}
Let $\mathit{Con}(G) \eqdef \{ G' ~|~ G \refines G'\}$,
$\mathit{Ref}(G) \eqdef \{ G' ~|~ G' \refines G \}$ be respectively the
set of contractions, refinements of a metric graph $G$.
\begin{enumerate}[label=(\roman*)]

  \item the refinement relation  $\refines$ is a partial order on the
    set of metric graphs;
    
  \item for any metric graph $G$, both $(\mathit{Con}(G), \refines)$
    and $(\mathit{Ref}(G), \refines)$ are complete lattices, moreover, 
    $(\mathit{Con}(G), \refines)$ is finite;
    

  \item for any metric graphs $G$, $G'$ if $G \refines G'$ then
      (1) the labelled transition systems $(\Pos{G}, \Segments, \gtran{G}{}{})$ and
      $(\Pos{G'}, \Segments, \gtran{G'}{}{})$ are strongly bisimilar and
      (2) the quasi-metric spaces $(\Pos{G}, \distance{G})$ and
        $(\Pos{G'}, \distance{G'})$ are isometric;
  \end{enumerate}
\end{proposition}
\begin{proof}
  (i) reflexivity, anti-symmetry and transitivity hold by definition of
  $\refines$ (ii) the least common contraction (resp. the
  greatest common refinement) graph is obtained by taking the intersection
  (resp. the union) of the sets of vertices and concatenating
  (resp. splitting) the set of edges accordingly.  Any contraction graph
  in $\mathit{Con}(G)$ is obtained from $G$ by removing a
  subset of its (finitely many) vertices with input and output degree equal
  to one (iii) the sets of positions and the set of rides are preserved by
  refinement, hence the labelled transition systems are bisimilar.
  Consequently, distances are preserved so the metric spaces are isometric.
\end{proof}

\subsubsection{Abstraction/Concretization.}

Consider $\Segments$, $\Segments'$ be sets of segments associated with
respectively concatenation $\cdot$, $\cdot'$, and length norm $\segnorm{.}$,
$\segnorm{.}'$.  A function $\alpha : \Segments \rightarrow \Segments'$ is a
\emph{segment abstraction} if it satisfies the following properties:
\begin{enumerate}[label=(\roman*)]
\item \emph{length preservation}: $\segnorm{s} = \segnorm{\alpha(s)}'$, for
  all $s \in \Segments$
\item \emph{homomorphism wrt concatenation}:
  $\alpha(s_1 \cdot s_2) = \alpha(s_1) \cdot' \alpha(s_2)$ for all
  $s_1, s_2 \in \Segments$ such that $s_1 \cdot s_2 \not= \bot$.
\end{enumerate}

For example, the function $\alpha^{CI}: \Segments_{curve} \rightarrow
\Segments_{interval}$ defined by $\alpha^{CI}( s ) \eqdef [0,
  \segnorm{s}]$ for all $s \in \Segments_{curve}$ is a an abstraction
of curve segments as interval segments.  Similarly, the function
$\alpha^{RC} : \Segments_{region} \rightarrow \Segments_{curve}$
defined by $\alpha^{RC}((s,w)) \eqdef s$ for all $(s,w) \in
\Segments_{region}$ is an abstraction of region segments as curve
segments.

Dually, we can define concretization functions $\gamma$ that go from
intervals to curves, and from curves to regions.  For example, for any
angles $\varphi$, $\theta$ consider
$\gamma^{IC}_{\varphi,\theta}:\Segments_{interval} \rightarrow
\Segments_{curve}$ where respectively,
$\gamma^{IC}_{\varphi,\theta}([0, a]) \eqdef arc[\frac{a}{\theta},
  \varphi, \theta]$ if $\theta \not= 0$ or
$\gamma^{IC}_{\varphi,\theta}([0, a]) \eqdef line[a,\varphi]$ if
$\theta=0$.  Or, for any positive real $w$ consider $\gamma^{CR}_w :
\Segments_{curve} \rightarrow \Segments_{regions}$ where
$\gamma^{CR}_w(s) \eqdef (s,w)$.

Given a segment abstraction $\alpha : \Segments \rightarrow
\Segments'$, a metric graph $G' = (V,\Segments',E')$ is an
$\alpha$-\emph{abstraction} of a metric graph $G=(V,\Segments,E)$,
denoted by $G' = \alpha(G)$, iff $G'$ is obtained from $G$ by
replacing segments $s$ by their abstractions $\alpha(s)$.  That is,
any edge $\gedge{G}{u}{s}{v}$ is transformed into an edge
$\gedge{G'}{u}{\alpha(s)}{v}$.  In a similar way,
$\gamma$-concretization on metric graphs is defined for a segment
concretization $\gamma : \Segments' \rightarrow \Segments$.

Fig.~\ref{fig:seg-abstraction} illustrates the use of the three
segment abstraction levels (respectively as intervals, curves,
regions) and their associated metric graphs.  Interval metric
graphs are $\alpha^{CI}$-abstractions of curve metric graphs,
which in turn are $\alpha^{RC}$-abstractions of region metric graphs.
Proposition \ref{prop:segment-refinement} states some key properties
on abstraction on metric graphs.

\begin{figure}[htbp]
  \centering
  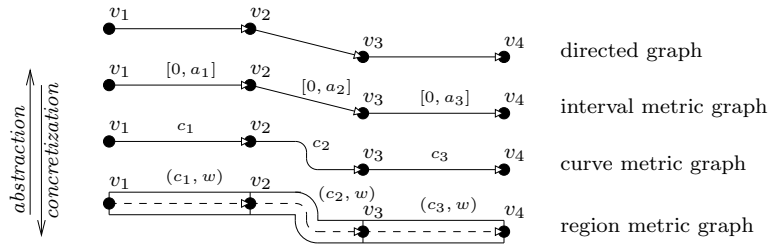
  \caption{\label{fig:seg-abstraction} Illustration of
    abstraction/concretization on metric graphs}
\end{figure}

\begin{proposition}\label{prop:segment-refinement}
  For a segment abstraction $\alpha : \Segments \rightarrow \Segments'$ and
  metric graphs $G$, $G'$ such that $G' = \alpha(G)$,
  the labelled transition system
  $(\Pos{G'}, \Segments', \gtran{G'}{}{})$ simulates the labelled transition
  system $(\Pos{G}, \Segments, \gtran{G}{}{})$ renamed by $\alpha$.
\end{proposition}
\begin{proof}
  The set of positions are preserved up to homomorphism by
  $\alpha$-abstraction and any $s$-labelled ride of $G$ can be
  simulated by an $\alpha(s)$-labelled ride of $G'$.  Nonetheless, the
  reverse is not necessarily true as segments $s_1$, $s_2$ that do not
  concatenate in $G$ may have abstractions $\alpha(s_1)$,
  $\alpha(s_2)$ that concatenate in $G'$ and therefore, leading to
  strictly more rides in $G'$ than $G$.
\end{proof}

\begin{proposition}\label{prop:segment-refinement-comm}
    Metric graph contraction and abstraction commute, that is, for any
    metric graphs $G$, $G'$, for any segment abstraction $\alpha$,
    if $G \refines G'$ then $\alpha(G) \refines \alpha(G')$, as
    depicted below.
\end{proposition}
\begin{figure}[h!]
  \hspace{1.5cm}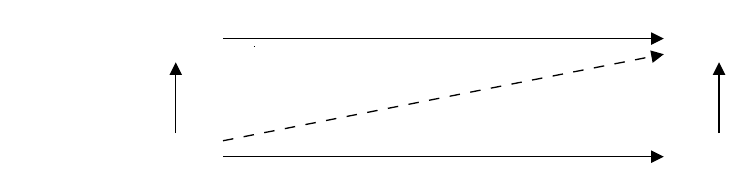
\end{figure}
\begin{proof}
  Recall that segment abstraction is an homomorphism wrt segment concatentation, that is, for any
  $s_1$, $s_2$ such that $s_1 \cdot s_2 \not= \bot$ it holds
  $\alpha(s_1 \cdot s_2) = \alpha(s_1) \cdot \alpha(s_2)$.  Therefore,
  contracting $s_1$ and $s_2$ into $s_{12}$ (from $G$ to $G'$) then abstracting
  $s_{12}$ as $\alpha(s_{12})$ (from $G'$ to $\alpha(G')$) is the same as abstracting
  $s_1$, $s_2$ into respectively $\alpha(s_1)$, $\alpha(s_2)$ (from $G$ to
  $\alpha(G)$) then contracting $\alpha(s_1)$ and $\alpha(s_2)$ into
  $\alpha(s_1) \cdot \alpha(s_2)$ (from $\alpha(G)$ to $\alpha(G')$).
\end{proof}

\subsubsection{2D-geometric consistency.}

The concept of \emph{2D-geometric consistency} is important when we use
metric graphs to represent 2D-maps.  A curve metric graph
$G=(V,\Segments_{curve},E)$ is 2D-geometrically consistent if it can be
embedded in the 2D-plane, that is, there exists an adressing mapping $\chi
: V \rightarrow \Reals^2$ of vertices to coordinates in the 2D-plane such
that for any edge $e = (v, s, v')$ it holds $\chi(v) + s(1) = \chi(v')$.
2D-geometric consistency can be checked by combining the following results.
\begin{proposition}\label{prop:consistency:acyclic}
  Any acyclic weakly connected curve metric graph is
  2D-geometrically consistent.
\end{proposition}
\begin{proof}
  Starting from the known coordinates $\chi(v_0)$ for some designated
  vertex $v_0$ it is possible to compute successively the coordinates
  of all the other vertices as between two vertices there exists only
  a single undirected path. So given the coordinates of any vertex of
  such a sequence we can compute the coordinates of any others such
  that to ensure consistency.
\end{proof}

\begin{proposition}\label{prop:consistency:cyclic}
  Any weakly connected curve metric graph is 2D-geometrically consistent iff for
  any elementary circuit $\omega$ (directed or undirected)  it holds $\sum_{e \in
    \omega^+} \segment{e}(1) = \sum_{e \in \omega^-} \segment{e}(1)$
  where $\omega^+$ (resp. $\omega^-$) denote the sets of edges taken
  with their direct (resp. reverted) orientation in $\omega$.
\end{proposition}
\begin{proof}
This is a necessary conditions because for given coordinates $\chi(v)$
of some vertex $v$ from which starts a circuit $\omega$,
the same coordinates are reached by following $\omega$, that is,
$\chi(v) = \chi(v) + \sum_{e \in \omega^+} \segment{e}(1) - \sum_{e \in
  \omega^-} \segment{e}(1)$ holds.  This condition is also sufficient
because if we remove edges of the graph so as to break all 
circuits, then by the previous proposition it is possible to find
consistent coordinate mappings.
\end{proof}

\begin{example}\label{ex:4-way-intersection}
  Consider the curve metric graph in Fig.~\ref{fig:intersection}
  representing a 4-way regular intersection with entrances $u_1$, $u_2$,
  $u_3$, $u_4$ and exits $v_1$, $v_2$, $v_3$, $v_4$.  The edges
  $e_{ij} = (\gedge{}{u_i}{s_{ij}}{v_j})$ are defined by the relations:
\[ \begin{array}{l l l} 
  \gedge{}{u_1}{line[2r+d, 0]}{v_3} & \gedge{}{u_1}{arc[r,0,-\frac{\pi}{2}]}{v_2} & \gedge{}{u_1}{arc[r+d,0,\frac{\pi}{2}]}{v_4} \\
  \gedge{}{u_2}{line[2r+d, \frac{\pi}{2}]}{v_4} & \gedge{}{u_2}{arc[r,\frac{\pi}{2},-\frac{\pi}{2}]}{v_3} & \gedge{}{u_2}{arc[r+d,\frac{\pi}{2},\frac{\pi}{2}]}{v_1} \\
  \gedge{}{u_3}{line[2r+d, \pi]}{v_1} & \gedge{}{u_3}{arc[r,\pi,-\frac{\pi}{2}]}{v_4} & \gedge{}{u_3}{arc[r+d,\pi,\frac{\pi}{2}]}{v_2} \\
  \gedge{}{u_4}{line[2r+d, -\frac{\pi}{2}]}{v_2} \hspace{1cm} & \gedge{}{u_4}{arc[r,-\frac{\pi}{2},-\frac{\pi}{2}]}{v_1} \hspace{1cm} & \gedge{}{u_4}{arc[r+d,-\frac{\pi}{2},\frac{\pi}{2}]}{v_3} 
\end{array} \]
Note that to prove the 2D-geometric consistency we need to check the
identities defined for elementary circuits as stated in
Proposition \ref{prop:consistency:cyclic}.  For instance, the circuit
visiting the vertices $u_2, v_4, u_1, v_3, u_2$ traverses forwards the edges
$e_{24}$, $e_{13}$ and backwards the edges $e_{14}$, $e_{23}$. 
Henceforth, one must check:
\[ \begin{array}{ccccccc}
\segment{e_{24}}(1) & + & \segment{e_{13}}(1) & = & \segment{e_{14}}(1)  & + & \segment{e_{23}}(1) \\
line[2r+d, \frac{\pi}{2}](1) & + & line[2r+d, 0](1) & = & arc[r+d,0,\frac{\pi}{2}](1) & + & arc[r, \frac{\pi}{2}, -\frac{\pi}{2}](1) \\
(0, 2r+d) & + & (2r+d, 0) & = & (r+d, r+d) & + & (r,r)
\end{array} \]
\end{example}
\begin{figure}[htbp]
  \begin{center}
    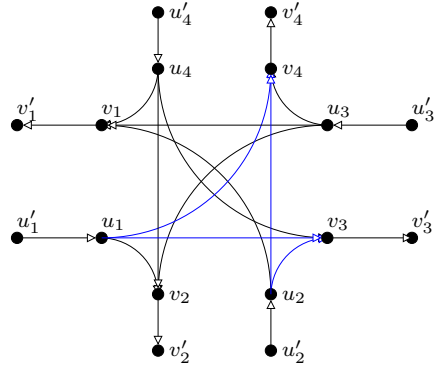
  \end{center}
  \caption{\label{fig:intersection}Curve metric graph representing a
    4-way intersection}
\end{figure}

\subsection{The Metric Configuration Logic}

\subsubsection{Syntax.}

Let consider a fixed set of segments $\Segments$ and assume there exists a
finite set $\Segments^T$ of segment constructors $s^T$ (or segment types),
that is, partial functions $s^T : \Reals^{m} \rightarrow \Segments^\bot$
for some natural $m$.  For example, we can take
$\Segments^T_{curve} = \{ line : \Reals^2 \rightarrow \Segments^\bot, arc :
\Reals^3 \rightarrow \Segments^\bot \}$ as the set of constructor curve
segments $\Segments_{curve}$.

Let $K$, $Z$, $X$ be distinct finite sets of \emph{variables} denoting
respectively reals, segments and vertices of a metric graph.  The
syntax of the \emph{metric configuration logic} (\MCL) is defined in
Table \ref{tab:mcl-syntax}.

\begin{table}[htbp]
$$ \begin{array}{rcll} \hline
  t & ::= & a \in \Reals ~|~ k \in K ~|~ t + t ~|~ t \cdot t & \mathit{arithmetic ~ terms} \\
\psi_K & ::= & t \le t' & \mathit{arithmetic ~ constraints} \\ [5pt]
s & ::= & s^T(t_1, ..., t_m) ~|~ z \in Z ~|~ s \cdot s & \mathit{segment ~ terms} \\
\psi_S & ::= & s = s' ~|~ \segnorm{s} = t &  \mathit{segment ~ constraints} \\[5pt]
p & ::= & x \in X ~|~ (x,s,t) ~|~ (t,s,x) & \mathit{position ~ terms} \\
     \psi_G & ::= & \gedge{}{x}{s}{x'} ~|~ p = p' ~|~ \gtran{}{p}{s}{p'}  ~|~ d(p,p') = t~~~
                                                             & \mathit{metric ~ graph ~ constraints} \\[5pt]
\phi & ::= & \psi_K ~|~ \psi_S ~|~ \psi_G  & \mathit{atomic ~ formula} \\
& | & \phi + \phi ~|~ \phi \vee \phi ~|~ \neg \phi & \mathit{non\textrm{-}atomic ~ formula} \\
& | & \exists k.~ \phi(k) ~|~ \exists z.~ \phi(z) ~|~ \exists x.~ \phi(x)  & \mathit{quantifiers} \\ \hline
  \end{array}$$
  \caption{\label{tab:mcl-syntax} \MCL\ Syntax}
\end{table}

\subsubsection{Semantics.}

Let $G = (V,\Segments,E)$ be a metric graph fixed in the context, and let
$\sigma$ be an assignment of variables $K$, $Z$, $X$ to respectively
reals $\Reals$, segments $\Segments$, vertices $V$.  As usual, we
extend $\sigma$ for evaluation of arithmetic terms (with variables
from $K$) into reals. Moreover, we extend $\sigma$ for the partial
evaluation of segment terms (with variables from $Z$) and position
terms (with variables from $Z$ and $X$) into respectively segments
$\Segments$ and positions $\Pos{G}$ as defined by the rules in Table
\ref{tab:mcl-terms-evaluation}.

\begin{table}[htbp]
$$\begin{array}{rclrcl} \\ \hline
    \sigma s^T( t_1, ..., t_m)  & \eqdef & s^T(\sigma t_1, ..., \sigma t_m)
  & \sigma (x,s,t) & \eqdef & \pos{G}{fwd}( \sigma x, \sigma s, \sigma t) \\
    \sigma ~s \cdot s' & \eqdef & \sigma s \cdot \sigma s'
  & ~~~~~ \sigma (t,s,x)  & \eqdef & \pos{G}{bwd}( \sigma x, \sigma s, \sigma t) 
\end{array}$$
\hfill where $\pos{G}{fwd}$,
$\pos{G}{bwd} : V \times \Segments \times \Reals \rightarrow
\Pos{G}^\bot$ are defined as \hfill ~
$$\begin{array}{rcll}
    \pos{G}{fwd}(v, s, a) & \eqdef & (e,a)
    & \mbox{ only if } \exists !~ e = (v, s, v') \in E,~ 0 < a < \segnorm{s} \\
    \pos{G}{bwd}(v, s, a) & \eqdef & (e, \segnorm{s} - a)
    & \mbox{ only if } \exists !~ e = (v', s, v) \in E,~ 0 < a < \segnorm{s} \\ \hline
\end{array}$$
\caption{\label{tab:mcl-terms-evaluation} Evaluation of \MCL\ terms}
\end{table}

We tacitly restrict to terms which evaluate
successfully in their respective domains.  The semantics of
the metric configuration logic is defined by the rules in
Table~\ref{tab:mcl-semantics}.  Note that a formula represents a
configuration of metric graphs sharing common characteristics. Besides
the logic connectives with the usual set-theoretic meaning, the
coalescing operator $+$ allows building graphs by grouping elementary
constituents characterized by atomic formulas relating positions via
segments. Hence, the formula $\phi_1+\phi_2$ represents the graph
configurations obtained as the union of configurations satisfying
$\phi_1$ and $\phi_2$ respectively. It differs from $\phi_1 \vee
\phi_2$ in that this formula satisfies configurations that satisfy
either $\phi_1$ or $\phi_2$.

\begin{table}[htbp]
$$\begin{array}{rclcl} \hline \sigma, G & \models & t \le t' & ~~\mbox{iff}~~ & \sigma t \le \sigma t' \\
    \sigma, G & \models & s = s' & \mbox{iff} & \sigma s = \sigma s' \\
    \sigma, G & \models & \segnorm{s} = t & \mbox{iff} & \segnorm{\sigma s} = \sigma t \\
    \sigma, G & \models & \gedge{}{x}{s}{x'} & \mbox{iff} & E= \{ (\sigma x, \sigma s, \sigma x') \} \\
    \sigma, G & \models & p = p' & \mbox{iff} & \sigma p = \sigma p' \\
    \sigma, G & \models & \gtran{}{p}{s}{p'} & \mbox{iff} & \gtran{G}{\sigma p}{\sigma s}{\sigma p'} \\
    \sigma, G & \models & d(p,p') = t & \mbox{iff} & \distance{G}(\sigma p, \sigma p') = \sigma t \\
    \sigma, G & \models & \phi_1 + \phi_2 & \mbox{iff} & \sigma, (V, E_1) \models \phi_1 \mbox{ and } \sigma, (V, E_2) \models \phi_2 \\
    & & & & \mbox{for some } E_1, E_2 \mbox{ such that } E_1 \cup E_2 = E \\
    \sigma, G & \models & \phi_1 \vee \phi_2 & \mbox{iff} & \sigma, G \models \phi_1 \mbox{ or } \sigma, G \models \phi_2 \\
    \sigma, G & \models & \neg \phi & \mbox{iff} & \sigma, G \not\models \phi \\
    \sigma, G & \models & \exists k.~ \phi & \mbox{iff} & \sigma[k \mapsto a], G \models \phi \mbox{ for some } a \in \Reals \\
    \sigma, G & \models & \exists z.~ \phi & \mbox{iff} & \sigma[z \mapsto s], G \models \phi \mbox{ for some } s \in \Segments \\
    \sigma, G & \models & \exists x.~ \phi & \mbox{iff} & \sigma[x \mapsto v], G \models \phi \mbox{ for some } v \in V \\ \hline
  \end{array}$$
  \caption{\label{tab:mcl-semantics}\MCL\ Semantics}
\end{table}

\subsubsection{Properties.} 

\begin{table}[htbp]
  $$\begin{array}{lrcl} \hline
    (A.i) & (\phi_1 + \phi_2) + \phi_3 & \equiv & \phi_1 + (\phi_2 + \phi_3) \\
    (A.ii) & \phi_1 + \phi_2 & \equiv & \phi_2 + \phi_1 \\
    (A.iii) & \phi + \mbox{false} & \equiv & \mbox{false} \\
    (A.iv) & \phi + \phi & \not\equiv & \phi ~~(\mathit{in ~  general}) \\
    (A.v) & \phi_1 + (\phi_2 \vee \phi_3) & \equiv & (\phi_1 + \phi_2) \vee (\phi_1 + \phi_3) \\[4pt]
    (B.i) & \closure{ \closure {\phi}} & \equiv & \closure{ \phi } \\
    (B.ii) & \phi & \implies & \closure{\phi} \\
    (B.iii) & \closure{(\phi_1 \vee \phi_2)} & \equiv & \closure{\phi_1} \vee \closure{\phi_2} \\
    (B.iv) & \closure{(\phi_1 + \phi_2)} & \equiv & \closure{\phi_1} + \closure{\phi_2} ~~\equiv~~
    \closure{\phi_1} \wedge \closure{\phi_2} \\[4pt]
    (C.i) & \gedge{}{x}{s}{x'} \wedge (\phi_1 + \phi_2) & \equiv & (\gedge{}{x}{s}{x'} \wedge \phi_1) + (\gedge{}{x}{s}{x'} \wedge \phi_2) \\
    (C.ii) & \mbox{true} & \equiv & (\gedge{}{x}{s}{x'} + \neg(\closure{ \gedge{}{x}{s}{x'}})) \vee
    \neg(\closure{ \gedge{}{x}{s}{x'}}) \\ [4pt]
    (D.i) & d(p,p')=t \wedge \gtran{}{p}{s}{p'} & \implies & t \le \segnorm{s} \\
    (D.ii) & d(p,p')=t \wedge d(p',p'')=t' & \implies & \exists k.~ d(p, p'') = k \wedge k \le t + t' \\ \hline
  \end{array}$$
  \caption{\label{tab:mcl-theorems}\MCL\ Theorems}
\end{table}

Table \ref{tab:mcl-theorems} provides a set of theorems giving insight into the
characteristic properties of the logic. Theorems $(A.i)\mbox{-}(A.v)$
illustrate important properties of the $+$ operator that is
associative and commutative but not idempotent. As explained below, of
particular interest for writing specifications are formulas of the
form $\closure{} \phi \eqdef \phi + \mathit{true}$. These are
satisfied by configurations with graphs that contain a subgraph
satisfying $\phi$. Hence, while the formula $\gedge{}{x}{s}{x'}$
characterizes the graphs with two vertices and a single edge labeled
by $s$, the formula $\closure{\gedge{}{x}{s}{x'}}$ characterizes the
set of graphs containing such an edge. Thus $\closure{}$ is a closure
operator which moreover satisfies theorems $(B.i)\mbox{-}(B.iv)$.  Finally,
theorems $(C.i)\mbox{-}(C.ii)$ relate the atomic formula
$\gedge{}{x}{s}{x'}$ to coalescing and the complement of their closure.
The two last theorems $(D.i)\mbox{-}(D.ii)$ differ from the others
in that they express specific properties of segment and
metric graph constraints.

\begin{proposition}
  Metric graph constraints not involving edge constraints of the form
  $\gedge{}{x}{s}{x'}$ are insensitive to metric graph contraction and
  refinement.
\end{proposition}
\begin{proof}
  This is an immediate consequence of Proposition
  \ref{prop:position-refinement}, point (iii), guaranteeing
  preservation of structural properties by contraction and
  refinement (bisimilarity, isometry).
\end{proof}
Note that stronger preservation results for (even simple fragments of) \MCL\ are hard to obtain because the domain of vertex variables is a  fixed set of vertices.
This  makes \MCL\ sensitive to both contraction and refinement.
For example, the formula $\exists x.~ \exists y.~ \gtran{}{x}{s}{y}$ may
not hold before and hold after refinement i.e., if a pair of vertices
$u$, $v$ satisfying the constraint is added by refinement.

We provide below abstraction preservation results for MCL formulas.  
Any segment abstraction $\alpha : \Segments \rightarrow \Segments'$ can be
lifted to segment terms by taking respectively $\alpha(s^T(t_1, ..., t_m))
\eqdef (\alpha s^T)(t_1, ..., t_n)$, $\alpha(s_1 \cdot s_2) \eqdef
\alpha(s_1) \cdot' \alpha(s_2)$, $\alpha(z) \eqdef z$.  Moreover, $\alpha$
can be further lifted to \MCL\ formulas on $\Segments$. We denote by
$\alpha(\phi)$ the \MCL\ formula on $\Segments'$ obtained by rewriting all
the segment terms $s$ occuring in $\phi$ by $\alpha(s)$.  The following
proposition relates abstractions on formulas to abstractions on metric
graphs.
\begin{proposition}
  Let $\phi$ be an existential positive \MCL\ formula.
  Then $G \models \phi$ implies $\alpha(G) \models \alpha(\phi)$ whenever:
  \begin{enumerate}[label=(\roman*)]
  \item $\phi$ does not contain distance constraints or
  \item for any connected edges $e_1$, $e_2$ such that $\post{e_1} =
    \pre{e_2}$ their segments compose, that is, $\segment{e_1} \cdot
    \segment{e_2} \not= \bot$.
  \end{enumerate}
\end{proposition}
\begin{proof}
  On one hand, the satisfiability of positive segment and graph constraints
  (except path constraints) are preserved by segment abstraction.  For
  example, if $s_1 = s_2$ holds then $\alpha(s_1) = \alpha(s_2)$ holds, if
  $\gedge{}{x}{s}{y}$ holds in $G$ then $\gedge{}{x}{\alpha(s)}{y}$ holds
  in $\alpha(G)$, etc.  On the other hand, path constraints rely on an
  implicit universal quantification over segments labeling acyclic rides in
  $G$. Therefore, they are preserved only if, the set of acyclic rides is
  not augmented by abstraction.  This is indeed the case if the condition
  (ii) holds.
\end{proof}

\subsubsection{Modeling styles.}

Configuration logics allow the characterization of configurations of graphs
by adopting two different and complementary styles.  The bottom-up style
consists in building graphs as the coalescing of atomic formulas specifying
connectivity relations between vertices. The top-down style consists in
giving the specification as the conjunction of formulas. Hence, the meaning
of the specification is the intersection of the meanings of its conjuncts.

\begin{figure}[htbp]
  \begin{center}
    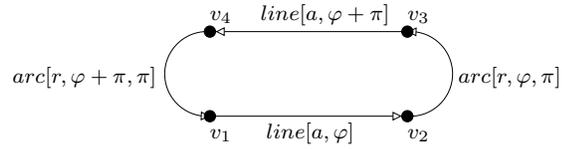
  \end{center}
  \caption{\label{fig:ring} A ring map}
\end{figure}

For example, consider the ring map of Fig.~\ref{fig:ring} with four
vertices and edges labeled by corresponding segments. The bottom-up
specification $\zeta_{ring}$ describes the ring as the
coalescing of the four segments as shown below (left side) whereas the
top-down specification $\xi_{ring}$ is the
conjunction of a set of rules whose application defines a ring map
(right side):
$$\begin{array}{l | c l}
  \exists a.~ \exists r.~ \exists\varphi.~ \exists x_1.~ \exists x_2.~ \exists x_3.~ \exists x_4.
  & & \exists a.~ \exists r.~ \forall \varphi.~\forall x_1.~\forall x_2.~\exists x_3 \\
  ~~\gedge{}{x_1}{line[a,\varphi]}{x_2} + \gedge{}{x_2}{arc[r,\varphi,\pi]}{x_3} ~+
  & & ~~\closure{\gedge{}{x_1}{line[a,\varphi]}{x_2}} \Rightarrow \closure{\gedge{}{x_2}{arc[r, \varphi, \pi]}{x_3}} ~\wedge \\
  ~~\gedge{}{x_3}{line[a,\varphi+\pi]}{x_4} + \gedge{}{x_4}{arc[r,\varphi+\pi,\pi]}{x_1}  
  & & ~~\closure{\gedge{}{x_1}{arc[r,\varphi, \pi]}{x_2}} \Rightarrow \closure{\gedge{}{x_2}{line[a, \varphi+\pi]}{x_3}} \\
\end{array}$$
It is possible to define a meaning-preserving correspondence between
top-down and bottom-up specifications as follows. We consider that
bottom up specifications are built by coalescing atomic formulas
$\zeta_k$ of the form $\gedge{}{x}{s}{x'}$.  Denote by $\xi_k$ the
formula $\xi_k \eqdef \closure{\zeta_k}$.

Given a formula $\sum_k \zeta_k$ the strongest corresponding top-down
specification is defined by the homomorphism: $\closure{\sum_k \zeta_k}
= \bigwedge_k \closure{\zeta_k} = \bigwedge_k \xi_k$. A weaker
top-down specification that admits an interesting interpretation as a
conjunction of implications is the formula $(\bigwedge_k \neg\xi_k)
\vee (\bigwedge_k \xi_k)$. Note that the conjunction $\bigwedge_k
\neg\xi_k$ is satisfied by models that do not contain any of the
segments of the bottom-up specification.  The following properties are
useful as they provide insight into to the way we can write top-down
weak specifications as a conjunction of “causality rules” expressed by
implications:
$$\begin{array}{c c c c c}
  \textstyle{\bigwedge_{k=1}^m} \neg \xi_k \vee \textstyle{\bigwedge_{k=1}^m} \xi_k & ~~\equiv~~ &
  \textstyle{\bigwedge_{1 \le k, k' \le m}} ~(\xi_k \Leftrightarrow \xi_{k'}) \\
  & \equiv & \textstyle{\bigwedge_{1 \le k, k' \le m}} ~(\xi_k \Rightarrow \xi_{k'}) & ~~\equiv~~ &
  \textstyle{\bigwedge_{k=1}^m} ~(\xi_k \Rightarrow \xi_{k \oplus 1})\end{array}$$
Note that $\bigwedge_k \xi_k = (\mathit{true} \Rightarrow \xi_1)
\wedge (\bigwedge_{k=1}^m \xi_k \Rightarrow \xi_{k \oplus 1})$.  As
an application of these results let us re-consider the example of
Fig.~\ref{fig:ring} and write the bottom-up specification
$\zeta_{ring}$ (making abstraction of quantifiers) as the
formula $\zeta_{12} + \zeta_{23} + \zeta_{34} + \zeta_{41}$.  The
corresponding weak top-down specification is then $(\xi_{12}
\Rightarrow \xi_{23}) \wedge (\xi_{23} \Rightarrow \xi_{34}) \wedge
(\xi_{34} \Rightarrow \xi_{41}) \wedge (\xi_{41} \Rightarrow
\xi_{12})$.  Note that the formula above can be simplified given that
the implications $\xi_{12} \Rightarrow \xi_{23}$, $\xi_{34}
\Rightarrow \xi_{41}$ and respectively $\xi_{23} \Rightarrow
\xi_{34}$, $\xi_{41} \Rightarrow \xi_{12}$ are of the same
form. Hence, they can be replaced by two parametric implications by
adequately changing quantification.  In that manner, we obtain
precisely the weak top-down specification corresponding to the
top-down specification $\xi_{ring}$.

\section{ADS Specification} \label{sec:ads-specification}

The results of the previous section provide a basis for the definition
of both a dynamic model for ADS and of logics for the expression of their
properties. The model is a timed transition system with states 
defined as the distribution of objects over of a metric graph
representing a map. Objects may be mobile such as vehicles and
pedestrians or static such as signaling equipment.  The logics are two
extensions of \MCL, one for the specification of predicates
representing sets of states and the other for the specification of
its behavior.

We introduce first the concept of map and its properties. Then we define
the dynamic model and the associated logics. Finally, we discuss the
validation problem and its possible solutions.

\subsection{Map specification}\label{sec:map-specification}

A weakly connected metric graph $G=(V,\Segments,E)$
can be interpreted as a map with a set of roads $R$ and a
set of junctions $J$, defined in the following manner:

\begin{itemize}
\item  a \emph{road} $r$ of $G$ is a maximal directed path
  $r=\gedge{G}{v_0}{s_1}{v_1}$, $\gedge{G}{v_1}{s_2}{v_2}$, ...,
  $\gedge{G}{v_{n-1}}{s_n}{v_n}$ where all the vertices $v_1$, ...,
  $v_{n-1}$ have indegree and outdegree equal to one.  We say that $v_0$ is
  the \emph{entrance} and $v_n$ is the \emph{exit} of $r$.  Let
  $R=\{r_i\}_{i\in I}$ be the set of roads of $G$.
\item a junction $j$ of $G$ is any maximal weakly connected sub-graph $G'$ of
  $G$, obtained from $G$ by removing from its roads all the vertices (and
  connecting edges) except their entrances and exits. Note that for a
  junction, its set of vertices of indegree one are exits of some road and
  its set of vertices of outdegree one are entrances of some road.  Let $J
  = \{j_\ell\}_{\ell \in L}$ be the set of junctions of $G$.
\end{itemize}

Note that $G$ is the union of the subgraphs representing its roads and
junctions. For every junction, the strong connectivity of $G$ implies
that from any entrance there exists at least one path leading to an
exit.  Additionally, we assume that maps include information about
features of roads, junctions that are relevant to traffic regulations:
\begin{itemize}
\item roads and junctions are \emph{typed}: road types can be highway,
  built-up area roads, carriage roads, etc.  Junctions types can be
  roundabouts, crossroads, highway exit, highway entrance, etc.
  We use standard notation associating a road or junction to its
  type e.g., $r: \mathit{highway}$, $j: \mathit{roundabout}$.
\item roads, junctions and their segments have \emph{attributes}. We
  use the dot notation $\att{a}{x}$ and $\att{a}{X}$ to denote
  respectively the attribute $x$ or the set of attributes $X$ of $a$.
  In particular, we denote by $\att{r}{en}$ and $\att{r}{ex}$
  respectively the entrance and the exit of a road $r$ and by $\att{j}{En}$
  and $\att{j}{Ex}$ the sets of entrances and exits of a junction
  $j$. Similarly, $\att{r}{lanes}$ is the number of lanes of the
  road $r$.
\end{itemize}

Note that contraction and refinement transform maps into maps. A road may be
refined into a road while a junction may be decomposed into a set of
roads and junctions. Furthermore, abstraction and concretization
transform maps into maps as they preserve their connectivity.

Given a map with sets of roads and junctions $R$ and $J$ respectively,
it is possible to derive compositionally its bottom-up and top-down
specifications.  We show first how we can get formulas $\zeta_j$,
$\zeta_r$ and $\xi_j$, $\xi_r$ for the bottom-up and top-down
specifications of $j$ and $r$, respectively. Let us consider the 
junctions illustrated in Fig.~\ref{fig:types-of-junctions}:

\begin{itemize}
\item if $ra$ is a roundabout with $n$ entrances $\att{ra}{En} =
  \{en_k\}_{k\in[1,n]}$ alternating with $n$ exits $\att{ra}{Ex} =
  \{ex_k\}_{k\in[1,n]}$ then its bottom-up
  specification is $\zeta_{ra} \eqdef \sum_{k=1}^n \zeta_k + \sum_{k=1}^n
  \zeta_{k,k+1}$, where $\zeta_k \eqdef \gedge{}{ex_k}{s_k}{en_k}$ and
  $\zeta_{k,k+1} \eqdef \gedge{}{en_k}{s_{k,k+1}}{ex_{k+1}}$. The top-down
  specification is $\xi_{ra} \eqdef \bigwedge_{k=1}^n \xi_k \wedge
  \bigwedge_{k=1}^n \xi_{k,k+1}$ where $\xi_k \eqdef \closure{\zeta_k}$
  and $\xi_{k,k+1} \eqdef \closure{\zeta_{k,k+1}}$.
  
\item if $in$ is an intersection with $n$ entrances $\att{in}{En} = \{
  en_k \}_{k=1,n}$ and $n$ exits $\att{in}{Ex} = \{ ex_k\}_{k \in
    [1,n]}$ then its bottom-up
  specification is $\zeta_{in} \eqdef \sum_{k=1}^n \zeta_k$ with $\zeta_k \eqdef
  \sum_{j\in J_k} \gedge{}{en_k}{s_{k,j}}{ex_j}$ and $J_k$ is the set
  of indices of the exits of $\att{j}{Ex}$ connected to the entrance
  $en_k$. Hence, the top-down specification is $\xi_{in} \eqdef
  \bigwedge_{k=1}^n \xi_k$ where $\xi_k \eqdef \closure{\zeta_k}$.
  
\item the formulas for a merger $mg$ and a fork $fk$ with respectively
  $n$ entrances and $n$ exits and unique exit and entrance
  respectively, can be obtained as a particular case of an
  intersection.
  
\item finally, for a road $r$ the corresponding specifications are
  $\zeta_r$ and $\xi_r \eqdef \closure{\psi_r}$ with $\zeta_r \eqdef
  \gedge{}{\att{r}{en}}{s_r}{\att{r}{ex}}$.
\end{itemize}

\begin{figure}[htbp]
  \centering
  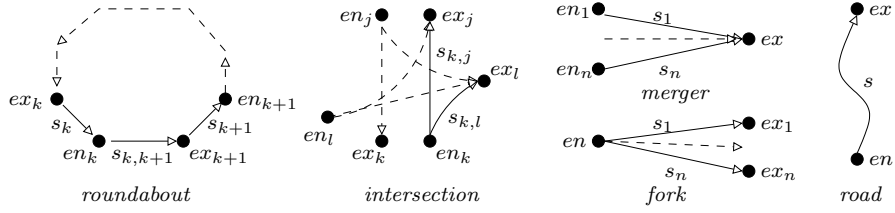
  \caption{\label{fig:types-of-junctions} Junctions and roads}
\end{figure}

Note that a map can be characterized by $R=\{ r_i \}_{i\in I}$, $J=\{
j_\ell \}_{\ell\in L}$, and a set of connectivity equations $E
\subseteq \{ r_i.en = j_\ell.ex_k ~|~ r_i\in R, j_\ell \in J, ex_k \in
\att{j_\ell}{Ex} \} \cup \{ r_i.ex = j_\ell.en_k ~|~ r_i\in R, j_\ell
\in J, en_k \in \att{j.\ell}{En} \}$ indicating how the road
entrances/exits are connected to junction exits/entrances (see
Fig.~\ref{fig:map-specification}). Let $\zeta_j$, $\xi_j$ and $\zeta_r$,
$\xi_r$ be the formulas corresponding to the bottom-up and top-down
specifications of a junction $j$ and a road $r$. Then the global map
specifications are $\zeta_M \eqdef \sum_{i\in I} \zeta_{r_i}
+ \sum_{\ell \in L} \zeta_{j_\ell} [E
/ \att{j_\ell}{En} \cup \att{j_\ell}{Ex}]$ where in $\zeta_{j_\ell}$ the
entrance and exit names of $j_\ell$ are replaced by the corresponding road
endpoints.

\begin{figure}[htbp]
  \centering
  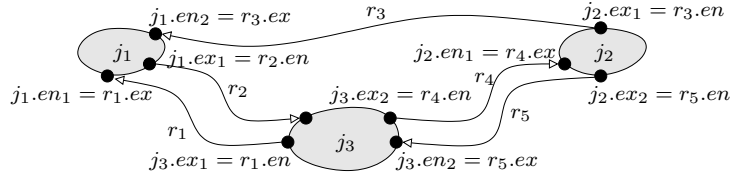
  \caption{\label{fig:map-specification}Map specification}
\end{figure}

\subsection{Dynamic ADS Model}\label{sec:dynamic-ads-model}

Given a metric graph $G$ representing a map, the state of an ADS is a tuple
$\vec{s} \eqdef \langle \vec{s}_o \rangle_{o \in \Objects}$
representing the distribution of a finite set of objects
$\Objects$ with their relevant dynamic attributes on the map $G$.  The
set of objects $\Objects$ includes a set of vehicles $\Vehicles$ and
sets of immobile equipment used to enforce traffic rules such as
lights, road signs, gates, etc.

For a vehicle $c$, its state $\vec{s}_c \eqdef \langle \mathit{it, pos, sp,
wt, ln, ...} \rangle$ includes respectively its \emph{itinerary} (from the
set of segments $\Segments$), its \emph{position} on the map (from
$\Pos{G}$), its \emph{speed} (from $\PosReals$), the \emph{waiting time}
(from $\PosReals$) which is the time elapsed since the speed of $c$ became
zero, the \emph{lane} it is traveling (from $\PosReals$), etc.  For a
traffic light $lt$, its state $\vec{s}_{lt} \eqdef \langle
\mathit{pos, cl, ...} \rangle$ includes respectively its $\emph{position}$ on the
map (from $\Pos{G}$), and its \emph{color} (with values 
\emph{red} and \emph{green}), etc.

For a given set of vehicles $\Vehicles$, we define below the transition
relation of the dynamic model on tuples of vehicle states
$\langle \vec{s}_c \rangle_{c \in \Vehicles}$ labeled by time increments
$\Delta t$.  For this purpose, we assume that each vehicle $c$ is equipped
with a function $\att{c}{ctrl}$ that determines its dynamics continuously
pursuing two goals: 1) keep the vehicle on the trajectory defined by its
itinerary; and 2) safety goals e.g., avoid collision with obstacles in its
neighborhood and respect traffic rules.   Then, the evolution of some key state variables is defined
as follows:
$$\begin{array}{rcl}
  \att{c}{sp}(t + \Delta t) & \eqdef & \att{c}{sp}(t) + \att{c}{ctrl}( \vec{s}(t) ) \cdot \Delta t \\
  \att{c}{pos}(t + \Delta t) & \eqdef & (\att{c}{pos}(t), \att{c}{it}(t), \att{c}{sp}(t) \cdot \Delta t) \\
  \att{c}{it}(t + \Delta t) & \eqdef & \subseg{\att{c}{it}}{ \att{c}{sp}(t) \cdot \Delta t}{ \segnorm{\att{c}{it}} }
\end{array} $$
That is, the vehicle $c$ travels at constant speed $\att{c}{sp}(t)$ for time $\Delta t$ and
\begin{enumerate}[label=(\roman*)]
\item its speed during the next interval $\Delta t$ is computed using the
  speed control function $\att{c}{ctrl}$ depending on system global state $s(t)$.
\item its next position is obtained from $\att{c}{pos}(t)$ following the itinerary
  $\att{c}{it}(t)$ for the distance $\att{c}{sp}(t) \cdot \Delta t$.
\item its next itinerary is obtained by erasing the initial sub-segment of
  the same distance from $\att{c}{it}(t)$.
\end{enumerate}

\begin{figure}[htbp]
  \centering
  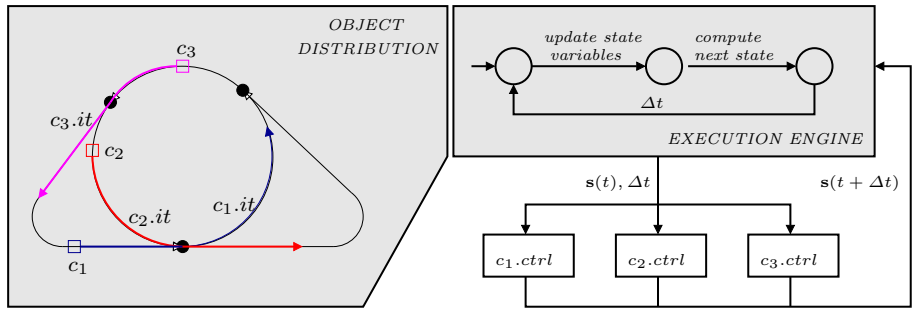
  \caption{\label{fig:mobile-systems-modeling} ADS
    Execution Principle}
\end{figure}

Fig.~\ref{fig:mobile-systems-modeling} depicts the execution principle by
an engine that cyclically updates the vehicle distribution on a map
representing their environment and coordinates the movement of components
representing each one a vehicle by choosing an adequate execution step
$\Delta t$. Of course, in this abstract execution we do not take into 
account various aspects of dynamism and reconfiguration discussed in
\cite{El-HokayemBS21}.

For a map $G$ and an initial state $\vec{s}^{(t_0)}$ we define a \emph{run}
as a sequence of consecutive states $[\vec{s}^{(t_i)}]_{i \ge 0}$
parameterized by an increasing sequence of time points $t_i \in \PosReals$,
equal to the sum of the time intervals elapsed for reaching the $i$-th state.

\subsection{Mobile \MCL\ and Scenario Description for ADS}

Mobile \MCL\ (shorthand \MMCL) is an extension of \MCL\ for the
specification of states of dynamic ADS models as distributions of objects
over maps.  It is equipped with object variables $Y$ with attributes
allowing to express constraints on object states.  Object variables in $Y$
are typed and denote objects from a finite set $\Objects$.  Constraints are
obtained by extending the syntax of \MCL\ to include object attribute
terms.  For example, if $y$ is a ''vehicle'' variable then $\att{y}{it}$ is
a segment term, $\att{y}{pos}$ is a position term, and $\att{y}{ln}$,
$\att{y}{sp}$, $\att{y}{wt}$ are arithmetic terms of \MMCL.
Moreover, \MMCL\ allows for equality $y = y'$ and existential
quantification $\exists y$ of object variables.

The semantics of \MMCL\ formulas is defined on distributions
$\langle \sigma, G, \vec{s} \rangle$ where $\sigma$ provides an
interpretation of variables (including object variables) to their
respective domains, $G$ is a metric graph representing the map, and
$\vec{s}$ is the system state vector for objects in $\Objects$.  The
evaluation of terms is extended to include object attributes, that is, for
any object variable $y$ with attribute $\mathit{attr}$ we define $\sigma
~\att{y}{attr} \eqdef \vec{s}_{\sigma y} (\mathit{attr})$.  Equality and
existential elimination on objects variables are interpreted with the usual
meaning, that is, $y = y'$ holds on $\langle \sigma, G, \vec{s} \rangle$
iff $\sigma y = \sigma y'$ and respectively $\exists y.~\psi$ holds on
$\langle \sigma, G, \vec{s} \rangle$ iff $\psi$ holds on
$\langle \sigma[y \mapsto o], G,
\vec{s} \rangle$ for some object $o \in \Objects$.

From a methodological perspective, we restrict to \MMCL\ formulas that can
be written as boolean combinations of three categories of sub-formulas:
\begin{enumerate}[label=(\roman*)]
\item $\psi_{map}$ describing map specifications
characterizing the static environment in which a dynamic system evolves,
\item $\psi_{dyn}$ describing relations between distributions of the objects of a dynamic system,
\item $\psi_{add}$ linking itinerary attributes of objects involved in $\psi_{dyn}$ to position addresses of maps described by $\psi_{map}$.
\end{enumerate}
The following set of primitives used respectively in sub-formulas of the
above categories is needed to express ADS scenarios and specifications:
\begin{enumerate}[label=(\roman*)]
\item for $x$, $x'$ vertex variables, $X$ set of vertex variables,
$[x ~\extname{right-of}~ x'~\extname{in}~X]$, $[x ~\extname{opposite}~ x'~\extname{in}~X]$
express constraints on the positioning of $x$, $x'$ with respect to the map restricted to vertices in $X$
(typically a junction): \\
$\small [x ~\extname{right-of}~x'~\extname{in}~X] \eqdef \exists a.\exists r.\exists \varphi.~
\bigvee_{x'' \in X}  \gedge{}{x'}{line[a,\varphi]}{x''} \wedge \gedge{}{x}{arc[r,\varphi+\frac{\pi}{2},-\frac{\pi}{2}]}{x''}$ \\
$\small [x ~\extname{opposite}~x'~\extname{in}~X] \eqdef \exists a.\exists \varphi.~
\bigvee_{x'',x''' \in X} \gedge{}{x}{line[a,\varphi]}{x''} \wedge \gedge{}{x'}{line[a,\varphi+\pi]}{x'''}$

\item for $c$, $o$ respectively vehicle, object variables, $d$ arithmetic
  term, $\meets{c}{d}{o}$ means that $c$ reaches the position of $o$
  at distance $d$:\\
$\small \meets{c}{d}{o} \eqdef \exists z, z'.~
\gtran{}{\att{c}{pos}}{z}{\att{o}{pos}} \wedge
\att{c}{it} = z\cdot z' \wedge \segnorm{z} = d$

\item a) for $c$ a vehicle variable, $X$ a set of vertex variables,
$[c~\extname{go-straight}~ X]$, $[c ~\extname{turn-right}~ X]$,
$[c~\extname{turn-left}~ X]$ express constraints on the itinerary of $c$ within the
map restricted to vertices in $X$ (typically, a junction): \\
$\small \begin{array}{lcl}[c ~\extname{go-straight}~ X] & \eqdef &
\exists a. \exists \varphi.\exists z.~
\bigvee_{x,x' \in X} \att{c}{pos} = x \wedge \gedge{}{x}{line[a,\varphi]}{x'} \wedge
\att{c}{it} = line[a,\varphi] \cdot z \\[3pt]
[c ~\extname{turn-right}~ X] & \eqdef & \exists r.\exists \varphi.\exists z.~
\bigvee_{x,x' \in X} \att{c}{pos} = x \wedge \gedge{}{x}{arc[r,\varphi,-\frac{\pi}{2}]}{x'}
\wedge \att{c}{it} = arc[r,\varphi,-\frac{\pi}{2}] \cdot z \\[3pt]
[c ~\extname{turn-left}~ X] & \eqdef & \exists r.\exists \varphi.\exists z.~
\bigvee_{x,x' \in X} \att{c}{pos} = x \wedge \gedge{}{x}{arc[r,\varphi,+\frac{\pi}{2}]}{x'}
\wedge \att{c}{it} = arc[r,\varphi,+\frac{\pi}{2}] \cdot z
\end{array}$

b) for $o$ an object variable, $X$ a set of vertex variables, $l$ an
  optional arithmetic term, $\inside{o}{X,l}$ means that the position
  of $o$ belongs to the map subgraph restricted to vertices in $X$ and
  the lane of $o$ is $l$: \\
$\small \inside{o}{X,l} \eqdef \left(
\exists d.\exists s.\bigvee_{x, x' \in X} \gedge{}{x}{s}{x'} \wedge
\att{o}{pos}=(x,s,d) \vee \att{o}{pos}=x\right) \!\wedge\! \att{o}{ln}=l$
\end{enumerate}

\subsubsection{Scenario Description for ADS.}

We define a scene as a triplet
$\langle \psi_{map}, \psi_{add}, \psi_{dyn}\rangle$ of \MMCL\ formulas
without universal quantifiers where $\psi_{add}$ defines the addresses of
the objects involved in $\psi_{dyn}$ in the map specified by $\psi_{map}$.
As for maps, a scene can have a top-down and a bottom-up specification
defined respectively by the formulas,
$\closure \psi_{map} \Rightarrow \psi_{add} \wedge \psi_{dyn}$ and
$\psi_{map} \wedge \psi_{add} \wedge \psi_{dyn}$.

A scenario is a sequence of scenes sharing a common map context and
intended to describe relevant partial states of an ADS run.  There are
several proposals for scenario description
languages \cite{ASAMOpenScenario-1.0.0,DammKMPR18,abs-1809-09310,abs-2010-06580}.
Figure \ref{fig:scene} presents a scenario of two scenes taken
from \cite{ASAMOpenScenario-1.0.0}.
The initial scene is defined by:
\[\small \begin{array}{rcl} 
    \psi_{map} & = & [r : road(x,s,y)] \wedge [ \att{s}{lanes}=2 ] \\
    \psi_{add} & = & \inside{ego}{r,1} \wedge \inside{c_1}{r,1} \wedge \inside{c_2}{r,2}\\
    \psi_{dyn} & = & \meets{ego}{84}{c_1} \wedge \meets{c_2}{100}{ego} \wedge
            [\att{ego}{sp} = \att{c_1}{sp}=100 \wedge \att{c_2}{sp}=110]
\end{array}\]
The second scene after the vehicle $c_2$ passes the $ego$ vehicle is:
\[\small \begin{array}{rcl}
    \psi'_{map} & = & [r : road(x,s,y)] \wedge [ \att{s}{lanes}=2 ] \\
    \psi'_{add} & = & \inside{ego}{r,1} \wedge \inside{c_1}{r,1} \wedge \inside{c_2}{r,1}\\
    \psi'_{dyn} & = & \meets{ego}{20}{c_2} \wedge \meets{c_2}{64}{c_1} \wedge 
          [\att{ego}{sp} = \att{c_1}{sp}=100 \wedge \att{c_2}{sp}=110]
\end{array}\]
Note that from a semantic point of view, a scene is characterized by
minimal models of \MMCL\ $\langle \sigma, G, \vec{s} \rangle$ that satisfy
the formula and where all irrelevant components of $\vec{s}$ are
omitted. For instance, in the minimal models of the two scenes only the
components of $\vec{s}$ corresponding to $c_1$, $c_2$ and $ego$ are taken.

\begin{figure}[htbp]
\centering
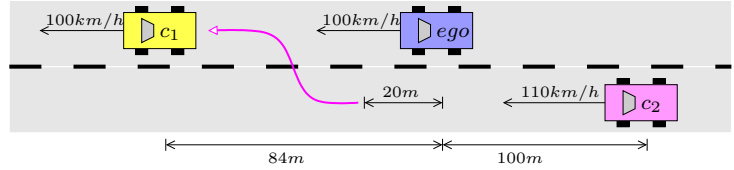
\caption{\label{fig:scene}A scenario}
\end{figure}

Of particular practical relevance are dynamic
scenarios \cite{abs-2010-06580,FremontKPSABWLL20} used to control the
execution of objects of an ADS model. These can be described as sequences
of guarded commands of the form
$sc \boldsymbol{\rightarrow} \att{c_1}{act_1} \& \cdots \& \att{c_n}{act_n}$
where the guard is a scene $sc$ involving a set of vehicles $c_1, \cdots
c_n$ and the command is a list of control actions to be executed by the
vehicles. The control action $\att{c}{act}$ is an action to be performed by
vehicle $c$ and affecting its own state variables. Hence,
$sc \boldsymbol{\rightarrow} \att{c_1}{act_1} \& \cdots \& \att{c_n}{act_n}$
describes a transition of an ADS from a scene $sc$ to a new scene obtained
by collateral modification of the state variables of the involved vehicles
after execution of the corresponding actions.  We define two control
actions for a vehicle $c$ and an interval constraint $cnt$ of the form
$speed_1 \le \att{c}{sp} \le speed_2$ :
\begin{itemize}
\item
$\att{c}{move} (cnt \mbox{ at } d)$ means that $c$ should travel to a
position at distance $d$ and $\att{c}{sp}$ satisfies $cnt$ at this
position. That is after executing $\att{c}{move} (cnt \mbox{ at } d)$
from state $sc$ with $\att{c}{pos} = p$ the new state will have
$\att{c}{pos}=p+d$ and the speed $\att{c}{sp}$ will satisfy $cnt$.
\item
$\att{c}{move}(cnt \mbox{ to } p)$ means that $c$ should travel to
position $p$ while $\att{c}{sp}$ continuously satisfies $cnt$. That is
after executing $\att{c}{move} (cnt \mbox{ at } d)$ from state $sc$ the
new state will have $\att{c}{pos} = p$ and the speed $\att{c}{sp}$
will satisfy satisfy $cnt$ at any position between the current
position and $p$.
\end{itemize}
We provide below few action rules enforcing ADS properties
and traffic rules:
\begin{enumerate}[label=(\roman*)]
\item $\meets{c}{d}{st} \boldsymbol{\rightarrow} \att{c}{move}(\mbox{0 at $d$})$ i.e., $c$ should brake to stop exactly before a stop sign;
\item $\inside{c}{\att{j}{En}}$ $\wedge$ $[\neg \exists c'. c'\not=c \wedge \inside{c'}{\att{j}{En}} \wedge [\att{c'}{wt} > \att{c}{wt}]]$ $\wedge$ [$\att{j}{ex} \in \att{c}{it}]$ $\boldsymbol{\rightarrow}$ $\att{c}{move}(\mbox{$v$ to $\att{j}{ex}$})$ i.e., if $c$ has been waiting longer than any other car at the entrance of a junction $j$, then it can move to the exit intersecting its itinerary reaching speed $v$;
\item $\inside{c}{j} \wedge [\att{j}{ex} \in \att{c}{it} ]\boldsymbol{\rightarrow} \att{c}{move}(\mbox{$v$ to $\att{j}{ex}$})$ i.e., if a vehicle $c$ is in a junction $j$ then it will move to its exit intersecting its itinerary;
\item $\meets{c}{d}{lt} \wedge [d \le d_{min}] \boldsymbol{\rightarrow} \att{c}{move}(\mbox{0 at $d$})$ i.e., it $c$ is approaching a light and the distance is less than a distance $d_{min}$ then $c$ should start braking to stop at $d$;
\item $\meets{c}{d}{lt} \wedge \inside{lt}{\att{j}{En}} \wedge [d\le d_{min}] \wedge [\att{lt}{cl} = \mbox{\it green}] \wedge [\att{j}{ex} \in \att{c}{it}] \boldsymbol{\rightarrow} \att{c}{move}(\mbox{$v$ to $\att{j}{ex}$})$
\end{enumerate}

\subsection{Temporal~\MMCL\ and Specification of ADS}

\emph{Temporal}-\MMCL\ (shorthand \TMMCL) is defined as the linear time
temporal extension of \MMCL.  The syntax is as follows:
$$ \Phi ::= \phi ~|~ \next \Phi ~|~ \Phi \until \Phi ~|~ \Phi \wedge \Phi ~|~ \exists c.~\Phi ~|~ \neg \Phi$$
where $\phi$ is \MMCL\ formula.  The semantics of \TMMCL\ is defined on
triples $(\sigma, G, [\vec{s}^{(t_i)}]_{i \ge 0})$ containing respectively
an assignment of vehicle variables defined in the \TMMCL\ context, a map
$G$ and a run $[\vec{s}^{(t_i)}]_{i \ge 0}$ on $G$ for a finite set of
objects $\Objects$.  The semantic rules are defined in Table
\ref{tab:tmmcl-semantics}.

\begin{table}[htbp]
  $$\begin{array}{rclcl}
    \sigma, G, [\vec{s}^{(t_i)}]_{i \ge 0} & \models & \phi & \mbox{iff} & \sigma, G, \vec{s}^{(t_0)}) \models \phi \\
    \sigma, G, [\vec{s}^{(t_i)}]_{i \ge 0} & \models & \next \Phi & \mbox{iff} & \sigma, G, [\vec{s}^{(t_i)}]_{i \ge 1} \models \Phi \\
    \sigma, G, [\vec{s}^{(t_i)}]_{i \ge 0} & \models & \Phi_1 \until \Phi_2 & \mbox{iff} &  \exists k\ge 0.~\forall j\in[0,k-1].~ \sigma, G, [\vec{s}^{(t_i)}]_{i \ge j} \models \Phi_1 \\
    & & & & \mbox{and } \sigma, G, [\vec{s}^{(t_i)}]_{i \ge k} \models \Phi_2 \\
    \sigma, G, [\vec{s}^{(t_i)}]_{i \ge 0} & \models & \Phi_1 \wedge \Phi_2 & \mbox{iff} &  \sigma, G, [\vec{s}^{(t_i)}]_{i \ge 0} \models \Phi_1 \mbox{ and } \sigma, G, [\vec{s}^{(t_i)}]_{i \ge 0} \models \Phi_2 \\
    \sigma, G, [\vec{s}^{(t_i)}]_{i \ge 0} & \models & \exists o.~\Phi & \mbox{iff} & \sigma[o \mapsto u], G, [\vec{s}^{(t_i)}]_{i \ge 0} \models \Phi, \mbox{ for some } u \in \Objects \\
    \sigma, G, [\vec{s}^{(t_i)}]_{i \ge 0} & \models & \neg \Phi & \mbox{iff} & \sigma, G, [\vec{s}^{(t_i)}]_{i \ge 0} \not\models \Phi   
   \end{array}$$
  \caption{\label{tab:tmmcl-semantics} Semantics of \TMMCL}
\end{table}

We use \TMMCL\ for both the specification of system properties
and traffic rules. The difference between the two concepts is not
clear-cut although it is implicit in many works. System properties
characterize the desired ADS behavior in terms of relations between
speeds and distances taking into account relevant dynamic
characteristics. These include properties such as keeping safe
distance or keeping acceleration and deceleration rates between some
bounds.

Traffic rules are higher-level specifications for enhanced safety and
efficiency that usually depend on the driving context. They deal not
only with obligations such as yielding right of way and traffic
control at junctions but also advice on how to drive sensibly and
safely in situations disrupting traffic flow such as congestion,
accidents and works in progress. We provide below a formalization of
system properties and traffic rules showing the expressiveness of our
modeling framework. We rather focus on the specification of traffic
rules because it is context sensitive and nicely illustrates the
intricacy of the combination of static and dynamic aspects.

Specifications related to junctions (i.e., intersections, roundabouts,
etc) are written as 
implications $\closure \zeta(j) \Rightarrow \Phi(j)$
where $\zeta(j)$ is a \MCL\ formula characterizing a junction $j$
(cf. section \ref{sec:map-specification}, see
example \ref{ex:4-way-intersection} for illustration) and $\Phi(j)$ is a \TMMCL\
formula specifying the temporal property holding in the context of
$j$.

\underline{System Properties:} We provide examples of general
properties that should be respected for vehicles for any type or road or junction.
\begin{enumerate}[label=(\roman*)]
\item keep safe distance, on the same lane or during overtaking: \\
  $\begin{array}{l}
  \forall c_1.~ \forall c_2.~ \always~ \left[ \forall d.~ \meets{c_1}{d}{c_2} \wedge [abs(\att{c_1}{ln} - \att{c_2}{ln}) < 1] \right. 
     \left. \Rightarrow [d \ge \mathit{safe\textrm{-}dist}(c_1,c_2)] \right]
  \end{array}$
  where $\mathit{safe\textrm{-}dist}(c_1, c_2)$ is
  the minimal distance for safe braking computed as a function of the
  speeds of the two vehicles/objects $c_1$, $c_2$ and the maximal braking force of $c_1$;
  
  \item reduce speed at proximity of a \emph{stop} sign: \\
  $\forall c.~ \forall st.~\always~ \left[ \forall d.~\meets{c}{d}{st}
    \Rightarrow [d \ge \mathit{safe\textrm{-}dist}(c,st)] \right]$

  \item whenever moving along arc segments, maintain the centrifugal force below some constant threshold $C$: \\
  $\forall c.~\always~ \left [ \forall r.~\forall \varphi.~\forall \theta.~\forall z'.~ [\att{c}{it} = arc[r,\varphi,\theta] \cdot z'] \Rightarrow [\att{c}{weight} \cdot (\att{c}{sp})^2 / r \le C] \right]$

    
\end{enumerate}

\underline{Traffic rules for intersection $j$ with all-way stop:} We formalize a set traffic rules provided in \cite{Wikipedia-All-way-stop}:
\begin{enumerate}[label=(\roman*)]
\setcounter{enumi}{3}
\item ``If a driver arrives at the intersection and no other vehicles are
present, then the driver can proceed'': \\
$\forall c.\forall st.~\always~ \inside{st}{\att{j}{en}} \wedge \inside{c}{\att{j}{en}} \wedge
[\neg \exists c'.~ c'\not=c \wedge \inside{c'}{j}] \Rightarrow \eventually \inside{c}{j}$

\item ``If, on approach of the intersection, there are one or more cars
already there, let them proceed, then proceed yourself'': \\
$\begin{array}{l}\forall c.\forall st.~\always~ \inside{st}{\att{j}{en}} \wedge \meets{c}{d}{st} \wedge [d \le d_{min}] \Rightarrow \\
\hspace{2cm} [\neg \inside{c}{j}] \until [\neg \exists c'. c'\not=c \wedge \inside{c'}{j}] \end{array}$ 


\item``If a driver arrives at the same time as another vehicle, the vehicle
on the right has the right-of-way'':\\
$\begin{array}{l}\forall c.\forall c'.~ \always~
\inside{c}{\att{j}{en}} \wedge \inside{c'}{\att{j}{en'}} \wedge
[\att{c}{wt} = \att{c'}{wt}] \wedge
[\att{j}{en} ~\extname{right-of}~ \att{j}{en'}~\extname{in}~j] \Rightarrow \\
\hspace{2cm} \inside{c'}{\att{j}{en'}} \until \inside{c}{j} \end{array}$

\item ``(a) If two vehicles arrive opposite each other at the same time, and
no vehicles are on the right, then they may proceed at the same time
if they are going straight ahead. (b) If one vehicle is turning and
one is going straight, the right-of-way goes to the car going
straight:'' \\
$\begin{array}{l}
\forall c.\forall c'.~ \always~
\inside{c}{\att{j}{en}} \wedge \inside{c'}{\att{j}{en'}} \wedge
[\att{c}{wt} = \att{c'}{wt} = 0] \wedge
[\att{j}{en} ~\extname{opposite}~ \att{j}{en'}~\extname{in}~j] \wedge \\ 
\hspace{1cm} \neg [\exists c''.~ \inside{c''}{\att{j}{en''}} \wedge
[\att{j}{en''} ~\extname{right-of}~ \att{j}{en} ~\extname{in}~ j] \vee
[\att{j}{en''} ~\extname{right-of}~ \att{j}{en'} ~\extname{in}~ j]] \wedge \\ 
\hspace{1cm} [c ~\extname{go-straight}~ j] \wedge [c' ~\extname{go-straight}~ j] ~\Rightarrow
~\eventually \inside{c}{j} \wedge \inside{c'}{j} \\[3pt]
\forall c.\forall c'.~ \always~
\inside{c}{\att{j}{en}} \wedge \inside{c'}{\att{j}{en'}} \wedge
[\att{c}{wt} = \att{c'}{wt} = 0] \wedge
[\att{j}{en} ~\extname{opposite}~ \att{j}{en'} ~\extname{in}~ j] \wedge \\ 
\hspace{1cm} \neg [\exists c''.~ \inside{c''}{\att{j}{en''}} \wedge
[\att{j}{en''} ~\extname{right-of}~ \att{j}{en}] \vee
[\att{j}{en''} ~\extname{right-of}~ \att{j}{en'}]] \wedge \\ 
\hspace{1cm} [c ~\extname{go-straight}~ j] \wedge \neg [c' ~\extname{go-straight}~ j] ~\Rightarrow~
\eventually \inside{c}{j}
\end{array}$

\item ``If two vehicles arrive opposite each other at the same time and one
is turning right and one is turning left, the right-of-way goes to the
vehicle turning right. Since they are both trying to turn into the
same road, priority should be given to the vehicle turning right as
they are closest to the lane'': \\
$\begin{array}{l}
\forall c.\forall c'.~ \always~
\inside{c}{\att{j}{en}} \wedge \inside{c'}{\att{j}{en'}} \wedge
[\att{c}{wt} = \att{c'}{wt} = 0] \wedge
[\att{j}{en} ~\extname{opposite}~ \att{j}{en'}~\extname{in}~j] \wedge \\ 
\hspace{1cm} [c ~\extname{turn-right}~ j] \wedge [c' ~\extname{turn-left}~ j] ~\Rightarrow~
\eventually \inside{c}{j}
\end{array}$
\end{enumerate}

\underline{Traffic rules for roundabout $j$:} We formalize two traffic rules provided in
\cite{WS-DT-roundabouts}:
\begin{enumerate}[label=(\roman*)]
\setcounter{enumi}{8}
\item ``Continue toward the roundabout and look to your left as you near the yield sign
and dashed yield line at the entrance to the roundabout. Yield to traffic already in the roundabout'':\\
$\forall c.~ \always~ \inside{c}{\att{j}{en}} \Rightarrow \inside{c}{\att{j}{en}} \until
\left[\neg \exists c'.\exists d.~ \inside{c'}{j} \wedge \meets{c'}{d}{c} \wedge [d\le d_{left}\right]]$
\item ``Once you see a gap in traffic, enter the circle and proceed to your exit. If there is no traffic in the roundabout, you may enter without yielding'':\\
$\forall c.~ \always~ \inside{c}{\att{j}{en}} \wedge [\neg \exists c'.\exists d.~ \inside{c'}{j} \wedge \meets{c'}{d}{c} \wedge [d \le d_{left}]]
\Rightarrow ~\eventually \inside{c}{j}$
\end{enumerate}

\underline{Traffic rules for intersection $j$ with traffic lights $lt$:}
\begin{enumerate}[label=(\roman*)]
\setcounter{enumi}{10}
  \item reduce speed at proximity of a traffic \emph{light}: \\
    $\forall c.~\forall lt.~\always~ \left[ \forall d.~ \meets{c}{d}{lt} \wedge [d \le d_{min}]
    \Rightarrow [ d \ge \mathit{safe\textrm{-}dist}(c,lt) ] \right]$
    
  \item eventually enter junction $j$ if the traffic \emph{light} on the entry is \emph{green}: \\
    $\begin{array}{l}
    \forall c.~\forall lt.~\always~ \left[ \forall d.~\meets{c}{d}{lt} \wedge [d \le d_{min}] \wedge \inside{lt}{\att{j}{En}} \right .
      \left. \wedge~ [\att{lt}{cl} = \mathit{green}] \Rightarrow \eventually \inside{c}{j} \right]     
  \end{array}$
\end{enumerate}

\section{ADS Validation} \label{sec:ads-validation}

%
%
%
%

\subsection{Classification of validation problems}

The following categories of validation
problems can arise in our framework:

\subsubsection{\MCL\ and \MMCL\ model-checking:} (i) given a map
  specification $\phi$ as a closed \MCL\ formula and a metric graph
  $G$ decide if $G$ is a model of $\phi$.  The problem boils down to
  checking satisfiability of a {\em segment logic} (\SL) formula
  obtained by quantifier elimination of vertex variables and partial
  evaluation of graph constraints in $\phi$ according to $G$.  We
  present later in this section a decision procedure for \SL.
  (ii) Similarly, given a distribution specification $\phi$ as a
  closed \MMCL\ formula, a map $G$ and a state $\vec{s}$ for a
  \emph{finite} set of objects $\Objects$, decide if $\langle
  G,\vec{s} \rangle$ is a model of $\phi$.  Again, the problem boils
  down to checking satisfiability of a \SL\ formula obtained by
  quantifier elimination of vertex and object variables an and partial
  evaluation of attribute terms.
  
\subsubsection{\TMMCL\ runtime verification:} given a temporal
  specification $\Phi$ as a \TMMCL\ formula, a map $G$ and a run
  $[\vec{s}^{(t_i)}]_{i \ge 0}$ of an ADS, check if
  $G,[\vec{s}^{(t_i)}]_{i \ge 0}$ is a model of $\Phi$.  This problem boils
  down to evaluating the semantics of $\Phi$ on the run.  In
  \cite{El-HokayemBS21} we consider a similar runtime verification problem
  for temporal configuration logic and runs of dynamic reconfigurable
  systems.  We have shown that the evaluation of linear-time temporal
  operators and the model-checking of state/configuration specifications
  can be dealt separately.  The same idea can be applied here: on one hand,
  the temporal formulas can be handled by LamaConv~\cite{LamaConv} to
  generate FSM monitors and on the other hand, the model-checking of
  distribution specifications can be handled by a SMT solver (such as Z3)
  by using an encoding into a decidable theory.

\subsubsection{\MCL\ and \MMCL\ satisfiability checking:} (i) given a
  map specification $\phi$ as closed \MCL\ formula decide if $\phi$ is
  satisfiable, that is, it has at least one model.  We show next in
  this section that the problem can be effectively solved for a
  significant fragment of \MCL\ including a restricted form of
  bottom-up map specifications.  Notice that entailment checking, that
  is, deciding validity of $\forall \vec{x}.~\phi_1 \Rightarrow
  \phi_2$ for map specifications $\phi_1$, $\phi_2$ where
  $fv(\phi_1)=fv(\phi_2)=\vec{x}$, boils down to checking
  satisfiability of $\exists \vec{x}.~\phi_1 \wedge \neg \phi_2$, and
  can be solved under the same restrictions.
  (ii) Similarly, given a distribution specification $\phi$ as a
  closed \MMCL\ formula decides if $\phi$ is satisfiable, that is, it
  has at least one model.  The problem can be reduced to the
  satisfiability checking of \MCL\ specifications whenever $\phi$ is
  of the restricted form $\exists y_1... \exists y_k.~ \phi'$ where
  $y_1$, ... $y_k$ are the only object variables occurring in $\phi'$.
  In this case, every object variable $y$ can be \emph{substituted} by
  a finite number of \MCL\ variables $y_{attr}$ encoding its identity
  and attributes.  As example, for a vehicle variable $y$ consider an
  identity (real) variable $y_{id}$, a segment variable $y_{it}$, a
  position variable $y_{pos}$, real variables $y_{ln}$, $y_{sp}$,
  $y_{wt}$, etc.  After replacement, we obtain an equisatisfiable
  \MCL\ formula by enforcing the additional constraints that state
  attributes are consistently assigned (e.g., $(y_{id} = y'_{id})
  \Rightarrow y_{it} = y'_{it})$ for all pairs $y$, $y'$ of vehicle
  variables among $y_1$, ..., $y_k$.  Finally, notice also that
  entailment checking between distributed specifications can be solved
  as well, by reduction to satisfiability checking as explained above.

%
%
%
%

\subsection{\label{subsec:mcl-satisfiability}Satisfiability checking}

\subsubsection{Satisfiability checking of \MCL.}

The satisfiability checking for \MCL\ formula is undecidable in
general.  Actually, the combined use of edge constraints $
\gedge{}{x}{s}{x'}$, equalities on vertex positions $x = x'$, boolean
operators and quantifiers leads to undecidability, as it allows the
embedding of first order logic on directed graphs.

Nevertheless, for a significant class of \MCL\ formula, their
satisfiability checking can be reduced to satisfiability checking of
\emph{segment logic} (\SL), that is, the fragment of \MCL\ without vertex
variables, which is a first order logic combining only arithmetic and
segment constraints.

A {\em complete metric graph specification} $\psi^*$ is a \MCL\ formula
of the form:
$$ \begin{array}{c}
  (\wedge_{1 \le i < j \le n}~ x_i \not= x_j) \wedge (\forall y. \vee_{i=1}^n ~ y = x_i ) ~ \wedge \\
  (\textstyle{\sum_{i=1}^n\sum_{j=1}^n\sum_{h=1}^{m_{ij}} ~ \gedge{}{x_i}{s_{ijh}}{x_j}}) \wedge
  (\wedge_{i=1}^n \wedge_{j=1}^n \wedge_{1 \le h < h' \le m_{ij}} ~ s_{ijh} \not= s_{ijh'})
  \end{array}$$

that is, where the set of free vertex variables is $\vec{x} = \{x_1,....,
x_n \}$.  Note that a complete metric specification characterizes precisely
a metric graph with precisely $n$ vertices (in correspondence with vertex variables
$x_1$, $\dots$, $x_n$) and, with precisely $m_{ij}$ distinct edges (that is,
defined by the constraints $\gedge{}{x_i}{s_{ijh}}{x_j}$ for $h=1,m_{ij}$), for every pair
of vertices $x_i$, $x_j$. 

\begin{theorem} \label{thm:tr-mcl-to-sl}
  Let $\psi^*$ be a complete metric graph specification with free
  variables $\vec{x} \uplus \vec{z} \uplus \vec{k}$.  For any
  \MCL\ formula $\phi$ with $fv(\phi) \subseteq \vec{x} \uplus \vec{z}
  \uplus \vec{k}$ it holds
  \begin{enumerate}
  \item the closed \MCL\ formula $\exists \vec{x}.~ \exists \vec{z}.~ \exists
    \vec{k}. ~ \psi^* \wedge \phi$ is satisfiable iff
  \item the closed \SL\ formula
    $$\begin{array}{c}\exists \vec{z}.~ \exists \vec{k}.~
    (\wedge_{i=1}^n \wedge_{j=1}^n \wedge_{1 \le h < h' \le m_{ij}} ~ s_{ijh} \not= s_{ijh'}) \wedge \\
    \hspace{3cm} (\wedge_{i=1}^n \wedge_{j=1}^n \wedge_{h=1}^{m_{ij}} ~ \segnorm{s_{ijh}} > 0) \wedge 
    tr(n, E^*, \mu^*, \phi)\end{array}$$ 
    is satisfiable, where $n = card ~\vec{x}$,
    $E^* = \cup_{i=1}^n \cup_{j=1}^n \{ (i, s_{ijh}, j)\}_{h=1,m_{ij}}$,
    $\mu^* = \{ x_i\mapsto i \}_{i=1,n}$ and the translation
    $tr(n,E,\mu,\phi)$ is defined in table \ref{tab:tr-mcl-to-sl}.
  \end{enumerate}
\end{theorem}
\begin{proof}
($1. \rightarrow 2.$) If the formula $\psi^* \wedge \phi$ is satisfiable,
then it has a metric graph model isomorphic to the (unique up to edge
labeling) metric graph $G_{\psi^*}$ specified by $\psi^*$.  The translated
formula $tr(n,E^*,\mu^*,\phi)$ represents the evaluation of the semantics
of $\phi$ on the metric graph $G_{\psi^*}$ according to the rules defined
in Table~\ref{tab:mcl-semantics}.  It must be therefore satisfiable as
well, as initially $\psi^* \wedge \phi$ is satisfiable. ($2.\Rightarrow
1.$) If the conjunction of the translated formula $tr(n,E^*,\mu^*,\phi)$
and the additional constraints has a model, ones can use it to build a
metric graph, isomorphic to $G_{\psi^*}$, satisfying both $\psi^*$ and
$\phi$.  In particular the additional constraints ensure that the metric
graph is well formed, that is, all edges are labeled by non-zero length
segments, and there are no replicated edges between any pairs of vertices.
\end{proof}

\begin{table}[htbp]
$$\begin{array}{rcl} \hline
  tr(n, E, \mu, \psi_K) & \eqdef & \psi_K \\
  tr(n, E, \mu, \psi_S) & \eqdef & \psi_S \\
  tr(n, E, \mu, \gedge{}{x}{s}{y}) & \eqdef &
  \left\{\begin{array}{ll} s = s_{ijh} & \mbox{if } E=\{(i,s_{ijh},j)\},~ \mu x = i,~ \mu y = j \\
                \mbox{\it false} & \mbox{otherwise} \end{array}\right. \\
  tr(n, E, \mu, p = p') & \eqdef & \predname{eq-pos}(n, E, \mu, p, p') \\
  tr(n, E, \mu, \gtran{}{p}{s}{p'}) & \eqdef & \predname{acyclic-path}(n, E, \mu, p, s, p') \\
  tr(n, E, \mu, \phi_1 + \phi_2) & \eqdef & \bigvee_{E_1 \cup E_2 = E} tr(n, E_1, \mu, \phi_1) \wedge tr(n, E_2, \mu, \phi_2) \\
  tr(n, E, \mu, \phi_1 \vee \phi_2) & \eqdef & tr(n, E, \mu, \phi_1) \vee tr(n, E, \mu, \phi_2) \\
  tr(n, E, \mu, \neg \phi) & \eqdef & \neg tr(n, E, \mu, \phi) \\
  tr(n, E, \mu, \exists k.~ \phi) & \eqdef & \exists k.~ tr(n, E, \mu, \phi) \\
  tr(n, E, \mu, \exists z.~ \phi) & \eqdef & \exists z.~ tr(n, E, \mu, \phi) \\
  tr(n, E, \mu, \exists x.~ \phi) & \eqdef & \bigvee_{i=1}^n tr(n, E, \mu[x \mapsto i], \phi) \\ \hline
  \end{array}$$
  \caption{\label{tab:tr-mcl-to-sl} Translation rules for Theorem \ref{thm:tr-mcl-to-sl} }
\end{table}
The complete definition of $\predname{eq-pos}$ and $\predname{acyclic-path}$
is provided in the appendix.  Also, notice that 
distance constraints of the form $\distance{}(p,p') = t$ have not been considered as they are
equivalent to
$$(p = p' \wedge t = 0) \vee ((\exists z.~ \gtran{}{p}{z}{p'}
\wedge \segnorm{z} = t) \wedge (\forall z.~ \gtran{}{p}{z}{p'}
\Rightarrow \segnorm{z} \ge t))$$

\subsubsection{Satisfiability checking of \SL.}
If segments $\Segments$ are restricted to particular interpretations,
the satisfiability checking of formula of \SL\ can be further reduced
to satisfiability checking of formula of extended arithmetics on reals.
\begin{theorem}\label{thm:tr-sl-to-arithmetic}
  If segments $\Segments$ are defined as intervals
  \begin{enumerate}
  \item the closed \SL\ formula $\phi$ is satisfiable iff
  \item the closed real arithmetic formula $tr_1(\phi)$ is satisfiable, where
    the translation $tr_1(\phi)$ is defined in Table \ref{tab:tr-sl-to-arithmetic}.
  \end{enumerate}
\end{theorem}
\begin{proof}
With interval interpretation, segments are precisely
determined by their length and all segment operations and constraints
boil down to operations and constraints on reals.  Moreover, we remark
that the transformation does not require
multiplication\footnote{except if needed for encoding the length of
  segment types} on real terms, henceforth, the translated formula
$tr_1(\phi)$ belongs to linear arithmetic iff all arithmetic
constraints $\psi_K$ within $\phi$ were linear.
\end{proof}

\begin{table}[htbp]
  $$\begin{array}{rcl rcl} \hline 
      tr_1(\segnorm{s^T(t_1,\cdots t_m)}) & \eqdef & len\mbox{-}s^T(t_1,...,t_m) &
            tr_1(s = s') &  \eqdef & tr_1(\segnorm{s}) = tr_1(\segnorm{s'}) \\
            
      tr_1(\segnorm{z}) & \eqdef & k_z &
      tr_1(\segnorm{s} = t) & \eqdef & tr_1(\segnorm{s}) = t \\

      tr_1(\segnorm{s \cdot s'}) & \eqdef & tr_1(\segnorm{s}) + tr_1(\segnorm{s'})) &
      ~~~ tr_1(\phi_1 \vee \phi_2) & \eqdef & tr_1(\phi_1) \vee tr_1(\phi_2) \\
      
      & & &
      tr_1(\neg \phi) & \eqdef & \neg tr_1 (\phi) \\

      & & &
      tr_1(\exists k.~ \phi) & \eqdef & \exists k.~ tr_1(\phi) \\

      & & &
      tr_1(\exists z.~ \phi) & \eqdef & \exists k_z.~ tr_1(\phi) \\ \hline
  \end{array}$$    
  \caption{\label{tab:tr-sl-to-arithmetic} Translation rules for Theorem \ref{thm:tr-sl-to-arithmetic} }
\end{table}

%
%
%
%

\subsection{Validation in the large - Catching up with the needs}

There is a big gap between the state of the art and needs for
validation of ADS achievable only through simulation and
testing due to overwhelming complexity. How existing results can be
integrated into a rigorous validation methodology intended to provide
conclusive trustworthiness evidence?  We bring elements of an
answer to the above question and show how the proposed framework
provides insight into the different aspects of validation and related
methodological and technical issues.

Fig.~\ref{fig:validation} depicts the architecture of a general
validation environment integrating three tools: i) a {\sc Simulator}; ii) a
dynamic {\sc Scenario Execution Engine}; and iii) a {\sc Monitor}. Simulation is
driven by actions generated by the {\sc Scenario Execution Engine}. It
consists in executing the model of some ADS and generating runs
checked online by the {\sc Monitor}. The three tools collaborate to
implement an efficient and rigorous validation methodology as
discussed below.

\begin{itemize}
\item
The {\sc Simulator} executes a dynamic ADS model obtained as the composition
of two entities: 1) a model of the world represented by maps and the
distribution of the objects with their kinematic attributes; 2)
behavioral models of objects and their possible interactions.
Following the execution principle presented in
\ref{sec:dynamic-ads-model}, the {\sc Simulator} exhibits cyclic
behavior alternating between concurrent actions of the objects for a
lapse of time and computation of the resulting state of the world.  We
assume that the {\sc Simulator} exports to the other tools implementations
of basic predicates, variables and actions via an interface e.g., in
the form of an API. The rigorous definition of the interface should
rely on a semantic model of the simulated system with well-defined
concept of state and execution step, as described in Section~\ref{sec:dynamic-ads-model}.

\item
The {\sc Scenario Execution Engine} drives the simulation process by
executing dynamic scenarios and providing sequences of actions
executed by vehicles in the {\sc Simulator}. Scenarios are chosen following
an adaptive test strategy \cite{BloemFGKPRR19}. The purpose is to lead
the simulated system to specific configurations e.g. to explore corner
cases and high-risk situations or to meet specific coverage criteria
as discussed below.

\item
The {\sc Monitor} continuously receives relevant state changes of the
global system behavior and applies run-time verification techniques
checking online that the ADS runs satisfy a given set of
specifications including traffic rules and specific
properties. Furthermore, the {\sc Monitor} provides diagnostics and KPIs
used by the {\sc Scenario Execution Engine} in its test strategy.
\end{itemize}

\begin{figure}[htbp]
\centering
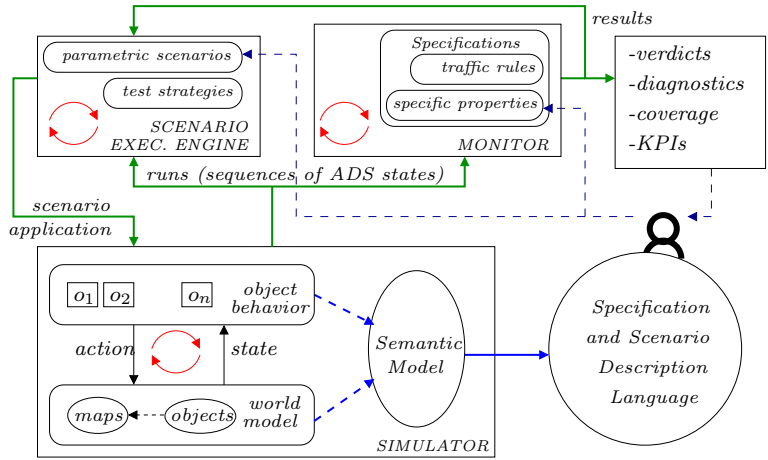
\caption{\label{fig:validation} Architecture of Validation Environment}
\end{figure}

Note that in the proposed architecture, scenarios should be consistent
with the behavior of the objects specified in the {\sc Simulator} and the
execution context. For instance, if a scenario requires that a vehicle
accelerate, this should be compatible with its collision avoidance
policy (if there is any). Furthermore, if the scenario enforces a left
turn for a vehicle this should be possible in the current execution
context. Many works focus on testing a particular “ego” vehicle with
respect to the collective behavior other mobile objects specified by
scenarios. Hence, while the behavior of the tested ego vehicle is
integrated in the {\sc Simulator}, the (often-rudimentary) behavior of the
other objects may ignore the particular execution context possibly
leading to inconsistency.

The described architecture is an adaptive testing environment. Using common
terminology, the {\sc Simulator} corresponds to a \emph{Unit Under Test},
the {\sc Monitor} is an \emph{Oracle} and the {\sc Scenario Execution
  Engine} plays the role of a \emph{Test Case Generator}
\cite{FremontKPSABWLL20}. To what extent is it possible to extend to ADS
conventional testing techniques? The application of a test strategy should
aim at optimizing criteria such as coverage and KPI’s. If for software
systems coverage can be defined as the ratio of source code exercised when
we run test sequences, it is not clear how a similar criterion could be
defined for ADS models.

To reduce the complexity of the space of situations to be explored by
test strategies, we need structuring criteria for scenarios.  One way
for achieving this aim is to define an equivalence relation on scenes
and by extension on scenarios preserving correctness for the
considered specifications. This idea that is also adopted by
metamorphic testing \cite{ZhouS19}, may allow a drastic reduction of
the test cases to be considered. Executions driven by equivalent
scenarios should be indistinguishable by the {\sc Monitor}. Furthermore,
testing one scenario per class could provide coverage for all its
scenarios.

The efficiency of exploration can be further improved if the
equivalence is defined by symmetrizing a risk pre-order on scenes and
by extension on scenarios. Intuitively, given a scene $sc$ and a set
of traffic rules $TR$ we can compute the set of active rules for $sc$,
i.e., the subset of rules of $TR$ that are applicable to $sc$. We can
consider the number of active rules for a scene $sc$ as a factor of
potential risk. Then for the same map context and the same car
distribution over roads and junction segments, a scene $sc$ is deemed
more risky than another $sc'$ if the set of active rules of $sc$
contains the set of active rules of $sc'$.

This idea can be further refined taking the approach applied in
\cite{El-HokayemBBS20,El-HokayemBS21} for the analysis of dynamic
reconfigurable systems. Given a set $TR$ of traffic rules, we can model the
risk of a scene $sc$ as a labeled hypergraph whose vertices are its objects
with their state attributes and the hyper-edges are traffic rules. In this
hypergraph, each object $o$ of $sc$ is connected to all the traffic rules
of $TR$ that are applicable to $o$.  The degree of $o$ measures somehow the
risk induced by interactions with other objects from the application of the
rules of $TR$. The risk hypergraph thus constructed for a given scene $sc$,
can provide a basis for the evaluation of the risk involved in $sc$ as a
function of its complexity and other state attributes.  We can define a
preorder relation on risk hypergraphs such that if a hypergraph $hg$
contains another hypergraph $hg'$ and all the kinematic attributes agree,
then the risk involved in the scene of $hg$ is higher.

Clearly, any risk hypergraph that is the union of disjoint hypergraphs
can be decomposed into hypergraphs that can be evaluated
independently. This often happens because of the locality of the
traffic rules; a rule is applicable to a set of objects that are
geographically close in a certain map context.

Hence, a test strategy could be progressive following the three steps:
\begin{enumerate}[label=(\roman*)]
\item
Consider scenarios involving scenes arising in simple traffic patterns
such as junctions and road types. Such scenarios could generate
sequences of scenes for roundabout, intersection, merger, overtaking
etc.
\item
For a given traffic pattern, guide the simulation to produce
higher-risk scenes e.g. by creating scenes with increasing number of
active traffic rules.
\item
Embed the tested simple high-risk traffic patterns in a map and
progressively apply scenarios leading to new higher-risk scenes.
\end{enumerate}

\section{Discussion} \label{sec:discussion}
The proposed framework relies on a minimal set of semantically integrated
concepts. It is expressive and modular as it introduces progressively the
basic concepts and carefully separates concerns. It supports a well-defined
specification and validation methodology without semantic gaps.  Using
configuration logic allows the specification of behavioral properties
taking into account map contexts. This is a main difference from approaches
relying on temporal logics that cannot account for map configurations and
where formulas characterize sets of runs in some implicit map environment,
usually a simple multi-lane setting. Configuration logic specifies scenes
as conjunctions of formulas describing map configurations and vehicle
distributions linked by an addressing relation. It enables enhanced
expressiveness by combining static and dynamic aspects while retaining the
possibility to consider them separately.  It considers maps as the central
concept of the semantic model and emphasizes the needs for multilevel
representation depending on the type of goals to be met including long-term
mission goals, mid-term maneuver goals and short-term safety and trajectory
tracking goals. Among the three abstraction levels, curve segment models
play a central role. Interval segment models can account for simple
properties depending only on relative distances between the involved
mobiles. For properties depending on topological and geometric relations,
curve segment models are needed. The expression of such properties involves
primitives such as go-straight, turn-right, turn-left, right-of and
opposite. Region segment models are needed for low level properties taking
into account the dimensions of the objects and their movement in the 2D
space.

The paper is the culmination of work developed over the past three years both
on foundations of autonomous systems \cite{HarelMS20,Sifakis19} and on
modelling and validation of reconfigurable dynamic systems using the DR-BIP
component framework \cite{BallouliBBS18,El-HokayemBS21}.  We plan to extend
this work in two directions. The first is to leverage on the DR-BIP
execution semantics and formalize ADS dynamics as the
composition of object behavior acting on maps. The second is to extend our
work on runtime verification of dynamic reconfigurable systems
\cite{El-HokayemBS21} by developing adaptive validation techniques driven by
adequate model coverage criteria. These techniques should provide
model-based evidence that a good deal of the many and diverse driving
situations are covered (e.g. different types of roads, of junctions, of
traffic conditions, etc). Finally, we will investigate
diagnostics generation techniques linking failures to their causes emerging
from risk factors such as violations of traffic regulations and
unpredictable events.

\appendix

\section{\MCL\ Satisfiability: Translation Rules}
\noindent $\predname{eq-pos}(n, E, \mu, p, p') \eqdef $
$$ \left\{\begin{array}{ll}
            
\mbox{\it true} & \mbox{if } p =x,~ p'=x', \mu x = \mu x' \\[5pt]
            
\multicolumn{2}{l}{s = s' \wedge t = t' \wedge 0 < t < \segnorm{s} \wedge \predname{unique}(s,S)} \\
& \mbox{if } p = (x,s,t),~ p'=(x',s',t'),~ \mu x = \mu x', \\
& ~~~ S = \{ s_{ijh} ~|~ (i, s_{ijh}, j) \in E,~ \mu x = i \} \\[5pt]

\multicolumn{2}{l}{s = s' \wedge t = t' \wedge 0 < t < \segnorm{s} \wedge \predname{unique}(s,S)} \\
& \mbox{if } p = (t,s,x),~ p'=(t',s',x'),~ \mu x = \mu x', \\
& ~~~ S = \{ s_{ijh} ~|~ (i, s_{ijh}, j) \in E,~ \mu x = j \} \\[5pt]

\multicolumn{2}{l}{s = s' \wedge t + t' = \segnorm{s} \wedge 0 < t < \segnorm{s} \wedge \predname{unique}(s,S) \wedge \predname{unique}(s',S')} \\ 
& \mbox{if } p = (x,s,t),~ p'=(t',s',x'), \\
& ~~~ S = \{ s_{ijh} ~|~ (i, s_{ijh}, j) \in E,~ \mu x = i \} \\
& ~~~ S' = \{ s_{ijh} ~|~ (i, s_{ijh}, j) \in E,~ \mu x' = j \} \\
& \mbox{/* and the symmetric case */} \\ [5pt]

\mbox{\it false} & \mbox{otherwise}
\end{array}\right. $$

Remark: in the above, $S$, $S'$ are multisets.

Remark: $x$ denotes a vertex position, $(x,s,t)$, $(t,s,x)$ denotes edge positions as $t$ must be strictly between 0 and $\segnorm{s}$

~

\noindent $\predname{unique}(s, \{s_1, ..., s_m \}) \eqdef (\vee_{h=1}^m s = s_h) \wedge
(\neg \vee_{1 \le h < h' \le m} s = s_h \wedge s = s_{h'})$

~

\noindent $\predname{acyclic-path}(n, E, \mu, p, s, p') \eqdef$
$$ \begin{array}{l}
\exists k. \exists k'. \exists z. \exists z'. \\
\bigvee_{e = (i,s_{ijh}, j) \in E} \predname{at-pos}(n, E, \mu, p, e, k) \wedge \predname{at-pos}(n, E, \mu, p', e, k') ~\wedge \\
\hspace{.5cm} (0 \le k \le k' \le \segnorm{s_{ijh}} \wedge \predname{subseg}(s_{ijh}, k, k', z) \wedge s = z) ~\vee \\[5pt]
     
\bigvee_{e=(i,s_{ijh}, j) \in E} \predname{at-pos}(n, E, \mu, p, e, k) \wedge \predname{at-pos}(n, E, \mu, p', e, k') ~\wedge \\
\hspace{.5cm} \big(0 \le k' \le k \le \segnorm{s_{ijh}} \wedge \predname{subseg}(s_{ijh}, k, \segnorm{s_{ijh}}, z) \wedge \predname{subseg}(s_{ijh}, 0, k', z') ~\wedge \\
\hspace{1cm} ((j = i \wedge s = z \cdot z') ~\vee \\
\hspace{1cm} (\bigvee_{w \in (E\setminus{e})^+_{ac}} ~\pre{w} = j \wedge \post{w}=i \wedge s = z \cdot \segment{w} \cdot z')) \big) \\[5pt]
     
\bigvee_{e = (i,s_{ijh}, j) \in E}\bigvee_{e'=(i',s_{i'j'h'},j') \in E\setminus{e}} \predname{at-pos}(n, E, \mu, p, e, k) \wedge \predname{at-pos}(n, E, \mu, p', e', k') ~\wedge \\
\hspace{.5cm} \big(0 \le k \le \segnorm{s_{ijh}} \wedge 0 \le k' \le \segnorm{s_{i'j'h'}} \wedge
\predname{subseg}(s_{ijh}, k, \segnorm{s_{ijh}}, z) \wedge \predname{subseg}(s_{i'j'h'}, 0, k', z) ~\wedge \\
\hspace{1cm} ((j = i' \wedge s = z \cdot z') ~\vee \\
\hspace{1cm} (\bigvee_{w \in (E\setminus(e,e'))^+_{ac}} ~\pre{w}=j \wedge \post{w}=i' \wedge s = z \cdot \segment{w} \cdot z')) \big)
\end{array} $$

~

\noindent $\predname{at-pos}(n,E,\mu, p, (i, s_{ijh}, j), k) \eqdef  $
$$\predname{eq-pos}(n, E, \mu, p, p) \wedge \left\{ \begin{array}{ll}
            k = 0 & \mbox{if } p = x,~ \mu x = i \\[5pt]
            k = \segnorm{s_{ijh}} & \mbox{if } p = x,~ \mu x = j \\[5pt]
            s = s_{ijh} \wedge k = t & \mbox{if } p = (x,s,t),~ \mu x = i \\[5pt] 
            s = s_{ijh} \wedge k + t = \segnorm{s_{ijh}} & \mbox{if } p = (t,s,x),~ \mu x = j \\[5pt] 
            \mbox{\it false} & \mbox{otherwise} 
          \end{array} \right.$$

~
        
\noindent $\predname{subseg}(s, t_1, t_2, s') \eqdef$
$$0 \le t_1 \le t_2 \le \segnorm{s} \wedge \exists z_1.~ \exists z_2.~ \segnorm{z_1} = t_1 \wedge \segnorm{z_2} + t_2 = \segnorm{s} \wedge s = z_1 \cdot s' \cdot z_2$$


\begin{thebibliography}{10}
\providecommand{\url}[1]{\texttt{#1}}
\providecommand{\urlprefix}{URL }
\providecommand{\doi}[1]{https://doi.org/#1}

\bibitem{OpenDRIVE-1.4}
{OpenDRIVE\textregistered\ Format Specification}. Tech. Rep. V~1.4
  \textcopyright 2006-2015, VIRES Simulationstechnologie GmbH (2015), retrieved
  from https://www.asam.net/standards/detail/opendrive

\bibitem{ASAMOpenDRIVE-1.6.0}
{ASAM OpenDRIVE\textregistered\ - Open Dynamic Road Information for Vehicle
  Environment}. Tech. Rep. V~1.6.0, ASAM e.V. (Mar 2020), retrieved from
  https://www.asam.net/standards/detail/opendrive

\bibitem{ASAMOpenScenario-1.0.0}
{ASAM OpenScenario\textregistered\ - Dynamic content in driving simulation, UML
  Modeling Rules}. Tech. Rep. V~1.0.0, ASAM e.V. (Mar 2020), retrieved from
  https://www.asam.net/standards/detail/openscenario

\bibitem{BagschikMM18}
Bagschik, G., Menzel, T., Maurer, M.: Ontology based scene creation for the
  development of automated vehicles. In: 2018 {IEEE} Intelligent Vehicles
  Symposium, {IV} 2018, Changshu, Suzhou, China, June 26-30, 2018. pp.
  1813--1820. {IEEE} (2018). \doi{10.1109/IVS.2018.8500632},
  \url{https://doi.org/10.1109/IVS.2018.8500632}

\bibitem{BallouliBBS18}
Ballouli, R.E., Bensalem, S., Bozga, M., Sifakis, J.: Four exercises in
  programming dynamic reconfigurable systems: Methodology and solution in
  {DR-BIP}. In: Margaria, T., Steffen, B. (eds.) Leveraging Applications of
  Formal Methods, Verification and Validation. Distributed Systems - 8th
  International Symposium, ISoLA 2018, Limassol, Cyprus, November 5-9, 2018,
  Proceedings, Part {III}. Lecture Notes in Computer Science, vol. 11246, pp.
  304--320. Springer (2018)

\bibitem{BeetzB18}
Beetz, J., Borrmann, A.: Benefits and limitations of linked data approaches for
  road modeling and data exchange. In: Smith, I.F.C., Domer, B. (eds.) Advanced
  Computing Strategies for Engineering - 25th {EG-ICE} International Workshop
  2018, Lausanne, Switzerland, June 10-13, 2018, Proceedings, Part {II}.
  Lecture Notes in Computer Science, vol. 10864, pp. 245--261. Springer (2018).
  \doi{10.1007/978-3-319-91638-5\_13},
  \url{https://doi.org/10.1007/978-3-319-91638-5\_13}

\bibitem{BloemFGKPRR19}
Bloem, R., Fey, G., Greif, F., K{\"{o}}nighofer, R., Pill, I., Riener, H.,
  R{\"{o}}ck, F.: Synthesizing adaptive test strategies from temporal logic
  specifications. Formal Methods Syst. Des.  \textbf{55}(2),  103--135 (2019)

\bibitem{ChenK18}
Chen, W., Kloul, L.: An ontology-based approach to generate the advanced driver
  assistance use cases of highway traffic. In: Aveiro, D., Dietz, J.L.G.,
  Filipe, J. (eds.) Proceedings of the 10th International Joint Conference on
  Knowledge Discovery, Knowledge Engineering and Knowledge Management, {IC3K}
  2018, Volume 2: KEOD, Seville, Spain, September 18-20, 2018. pp. 73--81.
  SciTePress (2018). \doi{10.5220/0006931700730081},
  \url{https://doi.org/10.5220/0006931700730081}

\bibitem{DammKMPR18}
Damm, W., Kemper, S., M{\"{o}}hlmann, E., Peikenkamp, T., Rakow, A.: Using
  traffic sequence charts for the development of {HAVs}. In: {ERTS} 2018,
  Toulouse, France, Jan 2018, Proceedings (2018)

\bibitem{DosovitskiyRCLK17}
Dosovitskiy, A., Ros, G., Codevilla, F., L{\'{o}}pez, A.M., Koltun, V.:
  {CARLA:} an open urban driving simulator. In: 1st Annual Conference on Robot
  Learning, CoRL 2017, Mountain View, California, USA, November 13-15, 2017,
  Proceedings. Proceedings of Machine Learning Research, vol.~78, pp. 1--16.
  {PMLR} (2017)

\bibitem{El-HokayemBBS20}
El{-}Hokayem, A., Bensalem, S., Bozga, M., Sifakis, J.: A layered
  implementation of {DR-BIP} supporting run-time monitoring and analysis. In:
  de~Boer, F.S., Cerone, A. (eds.) Software Engineering and Formal Methods -
  18th International Conference, {SEFM} 2020, Amsterdam, The Netherlands,
  September 14-18, 2020, Proceedings. Lecture Notes in Computer Science, vol.
  12310, pp. 284--302. Springer (2020)

\bibitem{El-HokayemBS21}
El{-}Hokayem, A., Bozga, M., Sifakis, J.: A temporal configuration logic for
  dynamic reconfigurable systems. In: Hung, C., Hong, J., Bechini, A., Song, E.
  (eds.) {SAC} '21: The 36th {ACM/SIGAPP} Symposium on Applied Computing,
  Virtual Event, Republic of Korea, March 22-26, 2021. pp. 1419--1428. {ACM}
  (2021)

\bibitem{EsterleAK19}
Esterle, K., Aravantinos, V., Knoll, A.C.: From specifications to behavior:
  Maneuver verification in a semantic state space. In: 2019 {IEEE} Intelligent
  Vehicles Symposium, {IV} 2019, Paris, France, June 9-12, 2019. pp.
  2140--2147. {IEEE} (2019). \doi{10.1109/IVS.2019.8814241},
  \url{https://doi.org/10.1109/IVS.2019.8814241}

\bibitem{EsterleGK20}
Esterle, K., Gressenbuch, L., Knoll, A.C.: Formalizing traffic rules for
  machine interpretability. In: 3rd {IEEE} Connected and Automated Vehicles
  Symposium, {CAVS} 2020, Victoria, BC, Canada, November 18 - December 16,
  2020. pp.~1--7. {IEEE} (2020). \doi{10.1109/CAVS51000.2020.9334599},
  \url{https://doi.org/10.1109/CAVS51000.2020.9334599}

\bibitem{abs-2010-06580}
Fremont, D.J., Kim, E., Dreossi, T., Ghosh, S., Yue, X.,
  Sangiovanni{-}Vincentelli, A.L., Seshia, S.A.: Scenic: {A} language for
  scenario specification and data generation. CoRR  \textbf{abs/2010.06580}
  (2020), \url{https://arxiv.org/abs/2010.06580}

\bibitem{FremontKPSABWLL20}
Fremont, D.J., Kim, E., Pant, Y.V., Seshia, S.A., Acharya, A., Bruso, X.,
  Wells, P., Lemke, S., Lu, Q., Mehta, S.: Formal scenario-based testing of
  autonomous vehicles: From simulation to the real world. In: 23rd {IEEE}
  International Conference on Intelligent Transportation Systems, {ITSC} 2020,
  Rhodes, Greece, September 20-23, 2020. pp.~1--8. {IEEE} (2020).
  \doi{10.1109/ITSC45102.2020.9294368},
  \url{https://doi.org/10.1109/ITSC45102.2020.9294368}

\bibitem{abs-1809-09310}
Fremont, D.J., Yue, X., Dreossi, T., Ghosh, S., Sangiovanni{-}Vincentelli,
  A.L., Seshia, S.A.: Scenic: Language-based scene generation. CoRR
  \textbf{abs/1809.09310} (2018), \url{http://arxiv.org/abs/1809.09310}

\bibitem{HarelMS20}
Harel, D., Marron, A., Sifakis, J.: Autonomics: In search of a foundation for
  next-generation autonomous systems. Proc. Natl. Acad. Sci. {USA}
  \textbf{117}(30),  17491--17498 (2020)

\bibitem{HilscherLOR11}
Hilscher, M., Linker, S., Olderog, E., Ravn, A.P.: An abstract model for
  proving safety of multi-lane traffic manoeuvres. In: Qin, S., Qiu, Z. (eds.)
  Formal Methods and Software Engineering - 13th International Conference on
  Formal Engineering Methods, {ICFEM} 2011, Durham, UK, October 26-28, 2011.
  Proceedings. Lecture Notes in Computer Science, vol.~6991, pp. 404--419.
  Springer (2011). \doi{10.1007/978-3-642-24559-6\_28},
  \url{https://doi.org/10.1007/978-3-642-24559-6\_28}

\bibitem{LamaConv}
{Institute for Software Engineering and Programming Languages, University of
  L\"ubeck}: {LamaConv} - {Logics and Automata Converter Library} (2020),
  \url{https://www.isp.uni-luebeck.de/lamaconv}

\bibitem{KarimiD20}
Karimi, A., Duggirala, P.S.: Formalizing traffic rules for uncontrolled
  intersections. In: 11th {ACM/IEEE} International Conference on Cyber-Physical
  Systems, {ICCPS} 2020, Sydney, Australia, April 21-25, 2020. pp. 41--50.
  {IEEE} (2020). \doi{10.1109/ICCPS48487.2020.00012},
  \url{https://doi.org/10.1109/ICCPS48487.2020.00012}

\bibitem{PoggenhansPJONK18}
Poggenhans, F., Pauls, J., Janosovits, J., Orf, S., Naumann, M., Kuhnt, F.,
  Mayr, M.: Lanelet2: {A} high-definition map framework for the future of
  automated driving. In: Zhang, W., Bayen, A.M., Medina, J.J.S., Barth, M.J.
  (eds.) 21st International Conference on Intelligent Transportation Systems,
  {ITSC} 2018, Maui, HI, USA, November 4-7, 2018. pp. 1672--1679. {IEEE}
  (2018). \doi{10.1109/ITSC.2018.8569929},
  \url{https://doi.org/10.1109/ITSC.2018.8569929}

\bibitem{RizaldiA15}
Rizaldi, A., Althoff, M.: Formalising traffic rules for accountability of
  autonomous vehicles. In: {IEEE} 18th International Conference on Intelligent
  Transportation Systems, {ITSC} 2015, Gran Canaria, Spain, September 15-18,
  2015. pp. 1658--1665. {IEEE} (2015). \doi{10.1109/ITSC.2015.269},
  \url{https://doi.org/10.1109/ITSC.2015.269}

\bibitem{RizaldiISA18}
Rizaldi, A., Immler, F., Sch{\"{u}}rmann, B., Althoff, M.: A formally verified
  motion planner for autonomous vehicles. In: Lahiri, S.K., Wang, C. (eds.)
  Automated Technology for Verification and Analysis - 16th International
  Symposium, {ATVA} 2018, Los Angeles, CA, USA, October 7-10, 2018,
  Proceedings. Lecture Notes in Computer Science, vol. 11138, pp. 75--90.
  Springer (2018). \doi{10.1007/978-3-030-01090-4\_5},
  \url{https://doi.org/10.1007/978-3-030-01090-4\_5}

\bibitem{RizaldiKHFIAHN17}
Rizaldi, A., Keinholz, J., Huber, M., Feldle, J., Immler, F., Althoff, M.,
  Hilgendorf, E., Nipkow, T.: Formalising and monitoring traffic rules for
  autonomous vehicles in {Isabelle/HOL}. In: Polikarpova, N., Schneider, S.A.
  (eds.) Integrated Formal Methods - 13th International Conference, {IFM} 2017,
  Turin, Italy, September 20-22, 2017, Proceedings. Lecture Notes in Computer
  Science, vol. 10510, pp. 50--66. Springer (2017).
  \doi{10.1007/978-3-319-66845-1\_4},
  \url{https://doi.org/10.1007/978-3-319-66845-1\_4}

\bibitem{abs-2005-03778}
Rong, G., Shin, B.H., Tabatabaee, H., Lu, Q., Lemke, S., Mozeiko, M., Boise,
  E., Uhm, G., Gerow, M., Mehta, S., Agafonov, E., Kim, T.H., Sterner, E.,
  Ushiroda, K., Reyes, M., Zelenkovsky, D., Kim, S.: {LGSVL} simulator: {A}
  high fidelity simulator for autonomous driving. CoRR  \textbf{abs/2005.03778}
  (2020), \url{https://arxiv.org/abs/2005.03778}

\bibitem{SchonemannWGO+19}
Sch{\"o}nemann, V., Winner, H., Glock, T., Otten, S., Sax, E., Boeddeker, B.,
  Verhaeg, G., Tronci, F., Padilla, G.G.: Scenario-based functional safety for
  automated driving on the example of valet parking. In: Arai, K., Kapoor, S.,
  Bhatia, R. (eds.) Advances in Information and Communication Networks. pp.
  53--64. Springer International Publishing, Cham (2019)

\bibitem{Sifakis19}
Sifakis, J.: Autonomous systems - an architectural characterization. In:
  Boreale, M., Corradini, F., Loreti, M., Pugliese, R. (eds.) Models,
  Languages, and Tools for Concurrent and Distributed Programming - Essays
  Dedicated to Rocco De Nicola on the Occasion of His 65th Birthday. Lecture
  Notes in Computer Science, vol. 11665, pp. 388--410. Springer (2019)

\bibitem{Wikipedia-All-way-stop}
Wikipedia: \url{https://en.wikipedia.org/wiki/All-way\_stop}

\bibitem{WS-DT-roundabouts}
{WS Dept. of Transportation}: \url{https://wsdot.wa.gov/Safety/roundabouts}

\bibitem{ZhouS19}
Zhou, Z.Q., Sun, L.: Metamorphic testing of driverless cars. Commun. {ACM}
  \textbf{62}(3),  61--67 (2019)

\end{thebibliography}
\end{document}